 \documentclass[final]{IEEEtran}
\usepackage{graphicx}% Include figure files
\usepackage{dcolumn}% Align table columns on decimal point
\usepackage{bm}% bold math
\usepackage{comment}

\usepackage{fancyhdr}
\usepackage{latexsym}
\usepackage{graphicx,color}
\usepackage{epstopdf}  %%after  \usepackage{graphicx,color}
\usepackage[pdfstartview=FitH]{hyperref}
\usepackage{xargs}
\usepackage{shadow,epsf,amsthm,amssymb,amsmath,caption,mathtools}
\mathtoolsset{showonlyrefs}

\usepackage{threeparttable}
\usepackage{multirow}
\usepackage{float}
\usepackage{makecell}
\usepackage{booktabs, tabularx, colortbl}
\usepackage{array}
\newcolumntype{C}{>{$}c<{$}}
\AtBeginDocument{
	\heavyrulewidth=.08em
	\lightrulewidth=.05em
	\cmidrulewidth=.03em
	\belowrulesep=.65ex
	\belowbottomsep=0pt
	\aboverulesep=.4ex
	\abovetopsep=0pt
	\cmidrulesep=\doublerulesep
	\cmidrulekern=.5em
	\defaultaddspace=.5em
}

\usepackage{enumerate}

\usepackage{setspace}                           %for doublespacing

\usepackage{lipsum}% just to automatically generate text

\usepackage{tikz}	
\usetikzlibrary{backgrounds,fit,decorations.pathreplacing}  % TikZ libraries

\usepackage[colorinlistoftodos,prependcaption]{todonotes}

\usepackage[ruled,vlined,noresetcount]{algorithm2e}
\SetKwInput{KwOutput}{Output}              % set the Output
\SetKwInOut{Input}{Input}
\SetKwInput{Output}{Output} 
\theoremstyle{plain}
\newtheorem{definitionenv}{Definition}
\newtheorem{lemmaenv}[definitionenv]{Lemma}
\newtheorem{theoremenv}[definitionenv]{Theorem}
\newtheorem{corollaryenv}[definitionenv]{Corollary}
\newtheorem{propositionenv}[definitionenv]{Proposition}
\newtheorem{conjectureenv}[definitionenv]{Conjecture}
\newtheorem{remarkenv}[definitionenv]{Remark}
\newenvironment{remark}{\begin{remarkenv}\rm}{\end{remarkenv}}
\newcommand{\br}{\begin{remark}}
	\newcommand{\er}{\end{remark}}

\newtheorem{exampleenv}{Example}
\newtheorem{app-lemmaenv}[section]{Lemma}

\newenvironment{definition}{\begin{definitionenv}\rm}{\end{definitionenv}}
\newenvironment{lemma}{\begin{lemmaenv}\rm}{\end{lemmaenv}}
\newenvironment{theorem}{\begin{theoremenv}\rm}{\end{theoremenv}}
\newenvironment{corollary}{\begin{corollaryenv}\rm}{\end{corollaryenv}}
\newenvironment{example}{\begin{exampleenv}\rm}{\end{exampleenv}}
\newenvironment{proposition}{\begin{propositionenv}\rm}{\end{propositionenv}}
\newenvironment{conjecture}{\begin{conjectureenv}\rm}{\end{conjectureenv}}
\newenvironment{app-lemma}{\begin{app-lemmaenv}\rm}{\end{app-lemmaenv}}

\newcommand{\bd}{\begin{definition}}
	\newcommand{\ed}{\end{definition}}
\newcommand{\bl}{\begin{lemma}}
	\newcommand{\el}{\end{lemma}}
\newcommand{\elp}{\hspace*{\fill} $\Box$
\end{lemma}}
\newcommand{\bt}{\begin{theorem}}
\newcommand{\et}{\end{theorem}}
\newcommand{\etp}{\hspace*{\fill} $\Box$
\end{theorem}}
\newcommand{\bc}{\begin{corollary}}
\newcommand{\ec}{\end{corollary}}
\newcommand{\ecp}{\hspace*{\fill} $\Box$
\end{corollary}}
\newcommand{\bcj}{\begin{conjecture}}
\newcommand{\ecj}{\end{conjecture}}

\newcommand{\be}{\begin{example}}
\newcommand{\ee}{\end{example}}
\newcommand{\eep}{\hspace*{\fill} $\Box$
\end{example}}
\newcommand{\bp}{\begin{proposition}}
\newcommand{\ep}{\end{proposition}}
\newcommand{\epp}{%\hspace*{\fill} $\Box$
\end{proposition}}

\newcommand{\bra}[1]{\langle#1|}
\newcommand{\ket}[1]{|#1\rangle}

\newcommand{\ketbra}[2]{|#1\rangle\langle#2|}
\newcommand{\tr}[1]{\text{tr}\left(#1\right)}

\newcommand{\eeq}{ \setcounter{equation} {\value{enumi}}}

\newcommand{\cG}{\mathcal{G}}
\newcommand{\cH}{\mathcal{H}}

\newcommand{\cO}{\mathcal{O}}
\newcommand{\cP}{\mathcal{P}}

\newcommand{\cS}{\mathcal{S}}

\newcommand{\bfc}{{\mathbf c}}

\newcommand{\bfi}{{\mathbf i}}
\newcommand{\bfj}{{\mathbf j}}

\newcommand{\bfr}{{\mathbf r}}
\newcommand{\bfs}{{\mathbf s}}

\newcommand{\bfv}{{\mathbf v}}

\newcommand{\bfw}{{\mathbf w}}

\newcommand{\mC}{{\mathbb C}}

\def\E{{\mathbb E}\,}

\def\tr{\textnormal{tr}}
\def\beq{\begin{equation}}
\def\eeq{\end{equation}}

\def\bean{\begin{IEEEeqnarray*}{rCl}}
\def\eean{\end{IEEEeqnarray*}}

\newcommand{\sgn}[1]{\textnormal{sgn}{\left(#1\right)}}

\newcommand{\ps}{{primary symplectic stabilizer}}
\newcommand{\pss}{{primary symplectic stabilizers}}

\newcommand{\iss}{{isotropic stabilizers}}
\newcommand{\Cl}{\textsf{Cl}}

\newcommand{\kT}{\ket{ \mathrm{+T}}}
\newcommand{\bT}{\bra{ \mathrm{+T}}}

\begin{document}
 
\title{  Learning quantum circuits of some $T$ gates}

\author{Ching-Yi Lai
           and~Hao-Chung Cheng 
           \thanks{ The work of	Ching-Yi Lai was supported by the Ministry of Science and Technology (MOST) in Taiwan, under Grant  MOST110-2628-E-A49-007.
           	The work of Hao-Chung Cheng  was supported partially from the Young Scholar Fellowship Program by MOST in Taiwan, under Grant MOST110-2636-E-002-009, and partially from the Yushan Young Scholar Program of the Ministry of Education in Taiwan, under Grant NTU-110V0904, and NTU-CC-111L894605.}
\thanks{Ching-Yi Lai is with the Institute of Communications Engineering, National Yang Ming Chiao Tung University, Hsinchu 30010, Taiwan (email:cylai@nycu.edu.tw).}% <-this % stops a space
\thanks{Hao-Chung Cheng is with the Department of Electrical Engineering, Graduate Institute of Communication Engineering, National Taiwan University (NTU), Taipei 10617, Taiwan, and with the Department of Mathematics, Institute of Applied Mathematical Sciences, NTU, Taipei 10617, Taiwan, and with the Center for Quantum Science and Engineering, NTU, Taipei 10617, Taiwan, and also with Hon Hai (Foxconn) Quantum Computing Centre, New Taipei City 236, Taiwan.}% <-this % stops a space
 }

\maketitle
 
\begin{abstract}
In this paper, we study the problem of learning an unknown quantum circuit of a certain structure. If the unknown target is an $n$-qubit Clifford circuit, we devise an efficient algorithm to reconstruct its circuit representation by using $O(n^2)$ queries to it. 
For decades, it has been unknown how to handle circuits beyond the Clifford group since  the stabilizer formalism cannot be applied in this case. Herein, we study quantum circuits of $T$-depth one  on the computational basis. 
We show that the output state of a $T$-depth one circuit  {\textit{of full $T$-rank}}  can be represented by a  stabilizer pseudomixture with a specific algebraic structure. Using Pauli and Bell measurements on copies of the output states, we can generate a hypothesis circuit that is equivalent to the unknown target circuit on computational basis states as input. 
If the number of $T$ gates of the target  is of the order $O({{\log n}})$, our algorithm requires  $O(n^2)$ queries to it and produces its equivalent circuit representation on the computational basis in time $O(n^3)$.
Using further additional $O(4^{3n})$ classical computations, we can derive an exact description of the target for arbitrary input states.
Our results greatly extend the previously known facts that stabilizer states can be efficiently identified based on the stabilizer formalism. 
\end{abstract}


\section{Introduction} 
\IEEEPARstart{C}{onsider} a quantum circuit $U$ that is accessible to us but its inner working or the mathematical description are unknown. Then, given a quantum state $\ket{\psi}$ as an input, what do we know about the circuit output?
The goal of this paper is   to find a circuit representation that resembles the functioning of $U$ so that we are able to predict the output  $U|\psi\rangle$ given an arbitrary input  $|\psi\rangle$.
Moreover, we would like to use as few queries to $U$ as possible.
We term such a problem \emph{learning an unknown quantum circuit}.

If we  focus on only the output state $U|\psi\rangle$ for a particular input state $|\psi\rangle$, the problem reduces to determining the state $U|\psi\rangle$.
This is also called \emph{quantum state tomography}, which is  one of the most crucial tasks in quantum information sciences \cite{Hra97, DPS03, BCG13, HWJ+17}.
Its goal is to \emph{infer} an unknown quantum state (assuming several copies of it are available) through   a sequence of quantum measurements  such that a proposed candidate state  \emph{performs well} in the future predictions.   
However, this is a non-trivial task.
In order to identify an unknown $n$-qubit quantum state, one would require exponentially many copies (in the number $n$) of the state during the tomography process to determine a full description of the state. This makes the task of tomography intractable in practice.

To mitigate such difficulties, at least two possible approaches are used.
Firstly, instead of fully characterizing the mathematical description of the unknown state, one might come up with a state that is \emph{probably approximately correct} (PAC) \cite{Val84} when only a particular set of measurements is of interest in the future predictions. In this case, Aaronson formulated the tomography of a quantum state as a learning problem and proved that it requires $O(n)$ copies of the state  to obtain a good hypothesis state \cite{Aar07} (see also \cite{CH16a, RAS+19, Aarr20, HK19, LC18, CD20} for the related works).
Secondly, one can focus on restricted states with a certain structure.
For example, stabilizer states are a class of states that play a substantial role in quantum error-correcting codes and other computational tasks \cite{NC00}. Aaronson and Gottesman provided a procedure to identify an unknown $n$-qubit stabilizer state with  $O(n)$ copies of it if collective measurements are possible~\cite{Got08}.
Later, Montanaro proposed an efficient algorithm via Bell sampling that consumes $O(n)$ copies of the state and runs in time of order $O(n^3)$ \cite{Mon17}.
Rocchetto cast the problem into the PAC learning model and showed that stabilizer states are \emph{efficiently PAC learnable} in the sense that the running time is polynomial in $n$ \cite{Roc18}.
(Note that the $O(n)$ number of copies is optimal by Holevo's theorem \cite{Hol73}.)

If now one aims to infer an unknown quantum evolution with certain known input states, this is called \emph{quantum process tomography}.
Once we completely know the  underlying evolution, we can determine the final states for arbitrary initial states. This is   the target problem we want to study in this paper.
Nevertheless, this problem is much more challenging than quantum state tomography.
The amount of resources needed for identifying an arbitrary $n$-qubit quantum circuit is $4^{3n}$, which is also practically formidable \cite{CN97, PCZ97, MRL08, SCP11}.
If a restricted class of quantum circuits are considered, such as \emph{Clifford circuits}   \cite{Got98}, 
Low showed that an $n$-qubit Clifford circuit $\mathsf{C}_n$ can be determined (up to a global phase) using $2n+1$ queries to $\mathsf{C}_n$ and $2n$ queries to its conjugate $\mathsf{C}_n^\dagger$ in time $O(n^2)$ \cite{Low09}. Moreover, a converse result showed that at least $n$ queries are required for such the task \cite{Low09}.
However, no concrete algorithms for reconstructing the circuit representation of the target $\mathsf{C}_n$ are provided, and one is not capable of predicting the output state of $\mathsf{C}_n$ on input an arbitrary quantum state. 
The first main contribution of this paper is to fulfill this gap.
Specifically, we propose a constructive algorithm to \emph{efficiently} produce the circuit representation of a target $n$-qubit Clifford circuit by using $O(n^2)$ queries to it. (See Theorem~\ref{theorem:Clifford} and Algorithm~\ref{Algorithm:I} in Section~\ref{sec:Clifford}.)
\begin{theorem}[Learning unknown Clifford circuits] \label{theorem:Clifford_intro}
	Given access to an unknown Clifford circuit $\mathsf{C}_n$, one can learn a circuit description  using  $2n^2+10n+4$ queries to it
	in time $O(n^3)$, so that the produced hypothesis circuit is equivalent to $\mathsf{C}_n$ with probability at least $1-2^{-n+1}$.
\end{theorem}
Let us emphasize that our approach does not require the access to the conjugate circuit $\mathsf{C}_n^\dagger$ (as it was required in \cite{Low09}). To implement such a conjugate  circuit might incur exponential overhead \cite{MSM19}, hence compromising the efficiency of a learning process. 
The problem of learning an unknown Clifford circuit is closely related to that of learning stabilizer states since the output of a Clifford circuit on input $\ket{0^n}$ is a stabilizer state.  In the proposed Algorithm~\ref{Algorithm:I}, we adopt Montanaro's Bell sampling algorithm for learning stabilizer states~\cite{Mon17} as a subroutine.
Moreover, we employ the \emph{stabilizer formalism} \cite{NC00, GK98} and exploit the desirable structure of the Pauli group to learn the output stabilizers states. 
Lastly, we show that a set of evolved Pauli basis can be identified by changing the input basis state appropriately
and thus we can determine the circuit representation of the unknown Clifford circuit $\mathsf{C}_n$ via a circuit synthesis procedure.

In the aforementioned task, we heavily rely on the stabilizer formalism.
However, it is unknown for a long time whether one can efficiently identify a quantum state that is produced from a quantum device beyond the class of Clifford circuits.
In this work, we  aim to provide an algorithm to identify the unknown quantum state  produced from a quantum circuit consisting of Clifford gates and a non-Clifford gate $T= \ket{0}\bra{0} + e^{i\pi/4} \ket{1}\bra{1}$.
This Clifford$+T$ gate set is universal for quantum computation and receives great attention  in  fault-tolerant quantum computation \cite{BK98,BK05}, compiling quantum circuits \cite{KMM13, Sel15, RS16}, and quantum circuit simulations \cite{BSS16, BG16}. 
An arbitrary quantum circuit can be approximately decomposed as a sequence of Clifford stages and $T$ stages, alternatively.
Here, a Clifford stage is simply a Clifford circuit.
In a $T$ stage, either a $T$ gate or the identity is applied to each qubit\footnote{We remark that $T^\dag$ can also be used here but $T$ and $T^\dag$ are equivalent up to a Clifford gate ($T^\dag=TS^\dag$). For our purpose,  we consider only $T$ gates.}.
The number of $T$ stages in a circuit is called the \textit{$T$-depth} of the circuit.
 
	A slightly related work is that	the quantum circuits in the Clifford hierarchy can be distinguished by some quantum measurements~\cite[Theorem 8]{Low09}.
However, no exact construction of the quantum measurements is given, other than the first level of Clifford hierarchy, namely the Pauli group, which can be identified using the idea of superdense coding~\cite{BW92}.
We note that quantum circuits of $T$-depth one have been studied in~\cite{Sel13}.

Our second main contribution is   as follows.  {
Consider an unknown $n$-qubit  $T$-depth one quantum circuit $U$ of full $T$-rank (see Def.~\ref{def:trank}) 
such that the $X$ part of its Pauli frame prior to the $T$ stage, restricted on the support of the $T$ gates, is of full rank. If $U$ has $O(\log n)$  $T$ gates,} we show that it requires at most $O(n^2)$ queries to efficiently learn the output of $U$ on input a computational basis state.

	\begin{theorem}[Learning unknown $T$-depth one circuits on the computational basis]
		Given access to an unknown $T$-depth one quantum circuit $U$   {of full $T$-rank}, one can learn a circuit description using  {$O(3^k n)$} queries to the unknown circuit $U$ with time complexity  {$O(n^3+3^kn)$}, where $k\leq n$ is the number of $T$ gates, so that
		the produced hypothesis circuit $\hat{U}$ is equivalent to $U$ with probability at least
		$1-3e^{-n}$ when the input states are restricted to the computational basis. 
	\end{theorem}

The explicit procedure is provided in Algorithm \ref{Algorithm:II} of Section~\ref{sec:learning_T}.
The reason why learning a quantum circuit $U$ of some $T$ gates is  much more technical-demanding  is elaborated in the following.
Since the $T$ gate does not belong to the Clifford group,  Pauli operators are not preserved under the evolution of  $U$, and hence the stabilizer formalism or the Gottesman--Kitaev theorem~\cite{GK98} does not work for circuits with $T$ gates. To circumvent such challenges, we need to conceive a scenario such that the Gottesman--Kitaev theorem can be leveraged for our purpose. 
Since a $T$ gate can be implemented by a $T$ gadget with the \emph{magic state} $\kT\bT$,
it yields two corresponding \emph{stabilizer pseudomixture} representations. Hence, a  quantum circuit can be transformed to a Clifford circuit with some ancillary magic states $\kT\bT$ and postselection~(see Proposition~\ref{prop:tgadget} of Section~\ref{sec:some_T}).
We remark that the stabilizer pseudomixtures have been studied in some contexts, such as robustness of magic~\cite{HC17}, classical and quantum simulation~\cite{BSS16,RLCK19}. 
In this work, we will exploit the structure of stabilizer pseudomixtures in learning quantum circuits.
We propose an \emph{expanded stabilizer formalism}, which serves as a crucial and convenient tool to analyze quantum circuits of some $T$ gates (see Section~\ref{sec:some_T}).
A $T$-depth one circuit can be represented by the expanded stabilizer formalism (Lemma~\ref{lemma:3tok} of Section~\ref{sec:some_T}).
In particular, this expanded stabilizer formalism can be generated by at most $2n$ \pss.
Consequently, we are able to analyze the outcomes of Pauli and Bell measurements on copies of its output state on input $\ket{0^n}$ (Lemmas~\ref{lemma:T1_state} and \ref{lemma:BellMeasure3} of Section~\ref{sec:learning_T}).
As a result, we recover a set of basis output states, which in turn reproduces a quantum circuit that agrees with the target circuit on  the computational basis (Theorem~\ref{theo:T} of Section~\ref{sec:learning_T}).
We remark that no conjugate oracle to $U$ is required in our algorithm. This is done after we carefully analyze the structure of such full-rank $T$-depth one output states (Lemma~\ref{lemma:Pauli_measurement_T} of Section~\ref{sec:learning_T}).

 Finally, we comment on the task of explicitly learning an unknown  {full $T$-rank} quantum unitary $U$ of $k$ $T$ gates at depth one, which generally requires an exponential number of computing resources. 
	Based on Algorithm~\ref{Algorithm:II}, we are able to find a description of $U\ket{\bfv}$ for $\bfv\in\{0,1\}^n$, using $O(3^k n)$ queries to $U$. Then we may compute  $U\ket{\bfv}\bra{\bfw}U^\dag$ for any $\bfv,\bfw$ to derive a full description of $U$,  {using $O(4^{3n})$ classical computations} (see Corollary~\ref{cor:tcircuit}). 	In contrast,  an exponential number of queries to $U$ are generally required in quantum process tomography.
   Our results are summarized in Table~\ref{table:summary}.

This paper is organized as follows. We provide preliminaries in Section~\ref{sec:preliminaries}. 
Section~\ref{sec:Clifford} is devoted to learning unknown Clifford circuits.
In Section~\ref{sec:some_T}, we propose an expanded stabilizer formalism for analyzing quantum circuits of some $T$ gates.
The learning algorithm for quantum circuits of full-rank $T$-depth one output state is given in Section~\ref{sec:learning}.
Then we conclude in Section~\ref{sec:conclusions}.

\begin{table*}[th!]
	\centering
	\begin{tabular}{c|c|c}  
		\hline
		\hline
		Target & Sample complexity & Time complexity \\
		\hline	
		Stabilizer states & $5n+2$ \cite{Mon17} & $O(n^3) $ \\
		
		\arrayrulecolor{gray!35}\hline
		
		Clifford circuits   & $2n^2+10n+4$   (Algorithm~\ref{Algorithm:I}) &  $O(n^3) $ \\

		\arrayrulecolor{gray!35}\hline
		
		full-rank $T$-depth one states   & $O(3^k n) $ (Algorithm~\ref{Algorithm:II}) & $O(n^3 +3^k n) $ \\
		
		\arrayrulecolor{gray!35}\hline
		
		 $T$-depth one circuits of full $T$-rank&  \multirow{2}{*}{ {$O(3^k n) $ Theorem~\ref{theo:T}}}  & \multirow{2}{*}{ {$O(n^3 +3^k n) $}} \\
		 {(on computational basis)} &  & \\
		
		\arrayrulecolor{gray!35}\hline
		
	  $T$-depth one circuits of full $T$-rank&  \multirow{2}{*}{ {$O(3^k n) $ Corollary \ref{cor:tcircuit}}}  & \multirow{2}{*}{ {$O(4^{3n}) $}} \\
		 {(on arbitrary inputs)} &  & \\
		
		\arrayrulecolor{gray!35}\hline
		
		$T$-depth hierarchies & $?$ & $?$ \\	
		\arrayrulecolor{black}\hline
		\hline		
	\end{tabular}
	\vspace{0.5em}
	\caption{Comparison of various $n$-qubit targets (and $k$ $T$ gates).		We only list the best upper bounds so far.
	}	\label{table:summary}
\end{table*}

\section{Preliminaries} \label{sec:preliminaries}

\subsection{Pauli operators and Clifford gates}
A pure quantum state, denoted by $\ket{v}$, is a unit vector in a certain Hilbert space.
Let $\{\ket{0},\ket{1}\}$ be an ordered basis (\textit{computational basis}) for pure single-qubit states in $\mC^2$. 
The Pauli matrices
\begin{align*} &\sigma_{00}=I=\begin{bmatrix}1 &0\\0&1\end{bmatrix},\ \sigma_{01}=X=\begin{bmatrix}0 &1\\1&0\end{bmatrix},\\
	&\sigma_{10}=Z=\begin{bmatrix}1 &0\\0&-1\end{bmatrix},\ \sigma_{11}=Y=\begin{bmatrix}0 &-i\\i&0\end{bmatrix}=iXZ
\end{align*}
form a basis of the space of linear operators $L(\mC^2)$. 
An important fact is that $X,Y,$ and $Z$ anticommute with each other.
Note that we may sometimes refer to $I$ as the identity matrix of appropriate dimension without ambiguity. 
The states $\ket{0}$ and $\ket{1}$ are the  eigenvectors of $Z$ and thus $\{\ket{0},\ket{1}\}$ is also called \textit{$Z$ basis}.
Similarly, the eigenvectors of $X $ are $\{\ket{+}=(\ket{0}+\ket{1})/\sqrt{2},$  $\ket{-}=(\ket{0}+\ket{1})/\sqrt{2} \}$, called \textit{$X$ basis},
and  the eigenvectors of $Y $ are $\{\ket{+i}=(\ket{0}+i\ket{1})/\sqrt{2},$ $\ket{-i}=(\ket{0}-i\ket{1})/\sqrt{2} \}$, called \textit{$Y$ basis}.

Associated with an $n$-qubit quantum system is a complex Hilbert space $\mC^{2^n}$.
A standard basis for linear operators on the $n$-qubit state space is the $n$-fold Pauli group, denoted by  $${\mathcal{P}}_n=\{c E_1\otimes  \cdots \otimes E_n: c\in\{\pm 1,\pm i\}, E_j\in\{I,X,Y,Z\} \}.$$
All the elements in ${\cP}_n$ are unitary with eigenvalues $\pm 1$ and they either commute or anticommute with each other. 
For convenience, we  may sometimes omit the symbol of tensor product $\otimes$.

An $n$-fold Pauli operator admits a binary representation that is irrelevant to its phase. For   $\bfw= w_1\cdots w_n,\bfv= v_1\cdots v_n\in \{0,1\}^n$, define
\[
\sigma_{\bfw:\bfv}\triangleq \bigotimes_{j=1}^n \sigma_{w_jv_j},
\]
where $\bfw:\bfv$ denote the concatenation of the two vectors $\bfw$ and $\bfv$.
Consequently, $\cP_n$ can be generated by $2n$ independent Pauli operators up to a phase in $ \{\pm 1,\pm i \}$.

The set of $n$-qubit \textit{Clifford circuits}, denoted by $\Cl_n$, consists of unitary operators that preserve the $n$-fold Pauli group under conjugation
\[
\Cl_n=\left\{V\in U(\mC^{2^n}):  V \cP_n V^{\dag}=\cP_n \right\}.
\]
Clifford circuits  are composed of Hadamard $H=\frac{1}{\sqrt{2}}\begin{bmatrix}1&1\\1&-1\end{bmatrix}$, 
phase $S=\begin{bmatrix}1&0\\0&i\end{bmatrix}$, and controlled-NOT $CNOT=\ketbra{0}{0}\otimes I+ \ketbra{1}{1}\otimes X$~(see, for example, \cite{NC00}) 
and these gates are called \emph{Clifford gates}.

We may use the notation $X_j$ to denote an operator that applies an $X$ to the $j$th qubit but trivially operates on the others.
For $\bfv\in\{0,1\}^n$, let $X^\bfv=  \prod_{i=1}^n X_{\bfv_i}$.
Other operators, such as $Y_j$, $Y^\mathbf{v}$, $S_j$, and $S^\bfv$, are similarly defined.
Thus 
any  $g\in {\mathcal{P}}_n$ can be expressed as $g= c Z^{\bfw}X^{\bfv}$ for some $ c \in \{\pm 1,\pm i \}$ and $\bfw,\bfv\in \{0,1\}^n$.

Without loss of generality, a basis of $\cP_n$ can be represented  as follows:
\begin{align}
	\begin{array}{cc}
		g_1= UX_1U^\dag,&  h_1= UZ_1U^\dag,\\
		g_2= UX_2U^\dag,&  h_2= UZ_2U^\dag,\\
		\vdots &\vdots\\
		g_n= UX_nU^\dag,&  h_n= UZ_nU^\dag,
	\end{array}  
\end{align}
where $U\in\Cl_n$ is a Clifford unitary. They satisfy the following commutation relations: \begin{align}
	\begin{array}{cc}
		g_ig_j= g_jg_i, &h_ih_j=h_jh_i,\\ g_ih_j=h_jg_i \mbox{ for }i\neq j, &g_ih_i= -  h_ig_i.
	\end{array}\label{eq:commutation}\end{align}
For a set of operators $\{g_1',\dots,g_n',h_1',\dots, h_n'\}$ satisfying the commutation relations~(\ref{eq:commutation}), the operators $g_i'$ and $h_i'$ are called    \textit{symplectic partners} of  each other.

\section{Learning unknown Clifford circuits}\label{sec:Clifford}
A Clifford circuit can be decomposed into Clifford gates in many ways.
This section is devoted to providing an efficient algorithm for finding a circuit representation of an unknown Clifford circuit.
In subsection~\ref{sec:stabilizer}, we recall how to retrieve information about an unknown stabilizer state by measurements. In subsection~\ref{sec:circuit_synthesis}, we describe a circuit synthesis method for Clifford circuits.
Lastly, in subsection~\ref{sec:learn_Clifford}, we add up the introduced tools to achieve our goal of learning an unknown Clifford circuit (Algorithm~\ref{Algorithm:I} and Theorem~\ref{theorem:Clifford}).

\subsection{Stabilizer states and measurements} \label{sec:stabilizer}
An $n$-qubit \emph{stabilizer state} $\ket{\phi}$ is the joint eigenvector of an Abelian group $\cS=\langle g_1,g_2,\dots, g_n\rangle \subset\cP_n$ that does not contain $-I^{\otimes n}$,
where $g_i\in\cP_n$ are independent  generators.
Any element $g$ in $\cS$ satisfies that $g\ket{\phi}=\ket{\phi}$ and is called a  \emph{stabilizer} of $\ket{\phi}$.
Suppose that $g_i= c_i \sigma_{\bfw_i:\bfv_i}$, where $c_i\in \{\pm1\}$ and $\bfw_i,\bfv_i\in\{0,1\}^n$, for $i=1,\dots,n$. Then 
a \textit{Pauli frame}   of the $n$-qubit stabilizer state $\ket{\phi}$ is given by
\[
\left(
\begin{array}{ccccc|cccc}
	\sgn{c_1}& w_{11}& w_{12}& \cdots & w_{1n}&v_{11}& v_{12}& \cdots & v_{1n}\\ 
	\sgn{c_2}& w_{21}& w_{22}& \cdots & w_{2n}&v_{21}& v_{22}& \cdots & v_{2n}\\ 
	\vdots &\vdots&\vdots&\ddots&\vdots &\vdots&\vdots&\ddots&\vdots \\
	\sgn{c_n}& w_{n1}& w_{n2}& \cdots & w_{nn}&v_{n1}& v_{n2}& \cdots & v_{nn}
\end{array}
\right),
\]
where $\sgn{c_i}=\pm$ for $c_i=\pm 1$.
\be
The two-qubit state $\left(\ket{00}-\ket{11}\right)/\sqrt{2}$ has a Pauli frame
$
\begin{pmatrix}
	-&0&0&1&1\\&1&1&0&0
\end{pmatrix}.
$
(The plus sign is omitted.)
\eep

\bd
A Pauli frame $ [\bfc\ W| V]$, where $\bfc\in\{+,-\}^n$, $W,V\in\{0,1\}^{n\times n}$ is said to have \textit{$X$-rank} $k$ if the binary matrix $V$ is of rank $k$. 
\label{def:xrank}.
\ed

It is known that the evolution of a Clifford circuit on a stabilizer state can be simulated by tracking the transformation of its Pauli frame according to the Gottesman--Kitaev theorem~\cite{GK98}.
Measuring a Pauli operator $g\in\cP_n$ on a stabilizer state $\ket{\phi}$ can also be tracked in a Pauli frame~\cite{NC00}.

An $n$-qubit stabilizer state $\ket{\phi}$ and its stabilizer group $\cS$ have a one-to-one correspondence. Hence $\ket{\phi}$  can be identified by its stabilizer group $\cS$. A stabilizer group can be described by a set of independent Pauli generators.
Assuming that many copies of $\ket{\phi}$ are available, we can   perform certain measurements on $\ket{\phi}$ and obtain a representation of $\cS$.
A naive method is to measure each Pauli operator $g\in\cP_n$ on a copy of $\ket{\phi}$.
The measurement returns outcome $+1$ with probability 1 if $g$ stabilizes $\ket{\phi}$
and returns outcome $+1$ $(-1)$ with probability $1/2$ ($1/2$), otherwise.
This requires $O(4^n)$ copies of the state for Pauli measurements.

Montanaro showed that Bell sampling on two copies of a stabilizer state returns one of its stabilizers, up to a Pauli operator that relates the stabilizer state and its \textit{conjugate state}~\cite[Lemma 2]{Mon17}.
Consequently  $O(n)$ outcomes of  Bell sampling on  pairs of the stabilizer state return $O(n)$ stabilizers of the state, which can then be used to determine a set of independent stabilizer generators with high probability.

For an arbitrary quantum state $\ket{\psi}=\sum_{\textbf{i}\in\{0,1\}^n}a_\mathbf{i}\ket{\mathbf{i}}\in \mathbb{C}^{2^n}$, its conjugate state is defined by 
\begin{align}
	\ket{\psi^*}=\sum_{\mathbf{i}}a_\mathbf{i}^*\ket{\mathbf{i}}.
\end{align} 
Inspired by Montanaro's Bell sampling~\cite[Lemma 2]{Mon17}, we reformulate the following lemma.
\bl \label{lemma:BellMeasure}
Suppose that $\ket{\psi}\in \mC^{2^n}$ is an $n$-qubit pure state. Then a joint Bell measurement on $\ket{\psi^*}\otimes \ket{\psi}$ returns outcome $\bfr\in\{0,1\}^{2n}$ with probability
\[
\frac{\left|\bra{\psi}\sigma_{\bfr}\ket{\psi}\right|^2}{2^n}.
\]

\el
\begin{proof}
	The outcome of a joint Bell measurement on  $\ket{\psi^*}\otimes \ket{\psi}$ is $\sigma_\bfr$
	with probability 
	\begin{align*}
		&\left| \bra{\Phi_+}^{\otimes n} (I^{\otimes n}\otimes \sigma_\bfr)  \ket{\psi^*}\ket{\psi}   \right|^2\\
		=&\frac{1}{2^n}\left| \sum_{\bfi\in\{0,1\}^n} \bra{\bfi}\otimes\bra{\bfi}  (I^{\otimes n}\otimes \sigma_\bfr)  \sum_{\bfj\in\{0,1\}^n} \alpha_\bfj^*\ket{\bfj}\otimes \ket{\psi}   \right|^2\\
		=&\frac{1}{2^n}\left| \bra{\psi}\sigma_\bfr \ket{\psi}   \right|^2.  
	\end{align*}
\end{proof}

Therefore,  
we can learn a stabilizer of a stabilizer state $\ket{\phi}$  by performing a joint Bell measurement on $\ket{\phi^*}\otimes \ket{\phi}$. 
If copies of the conjugate state are not available, stabilizers can be still learned by Bell measurements on $\ket{\phi}\otimes \ket{\phi}$  as in the following corollary.
\bc
Suppose that $\ket{\phi}\in \mC^{2^n}$ is an $n$-qubit stabilizer state.
Then there exists $\bfr_0\in\{0,1\}^{2n}$ such that
a joint Bell measurement on  $\ket{\phi}\otimes \ket{\phi}$ returns outcome $ \bfr$
with probability 
\[
\frac{\left|\bra{\phi}\sigma_{\bfr}\ket{\phi^*}\right|^2}{2^n}=\frac{\left|\bra{\phi}\sigma_{\bfr\oplus \bfr_0}\ket{\phi}\right|^2}{2^n}.
\]
\label{cor:Bell_stab}
\ec
\begin{proof}
	Suppose that $\ket{\phi}$ is stabilized by independent stabilizer generators $g_1,\dots,g_n\in\cP_n$.
	It is straightforward to see that the conjugate state $\ket{\phi^*}$ is stabilized by  $g_1^*,\dots,g_n^*$, 
	where $g_i^* = g_i$ if the number of its Pauli component $Y$  is even, and $g_i^* =- g_i$, otherwise.
	If $g_i^*=g_i$ for all $i=1,\dots,n$, then $\ket{\phi^*}=\ket{\phi}$ and the statement holds trivially.
	
	Now assume that $g_i^* =- g_i$ for some $g_i$, say $g_1$ for convenience, and
	$g_i^* =  g_i$ for $i=2,\dots, n$. This can be done because  if $g_j^* =- g_j$ for $g_j\neq g_1$,  we can replace it by  $g_j'= g_1g_j$ such that $(g_j')^*=g_j'$.
	Hence, the conjugate state $\ket{\phi^*}$ is stabilized by $-g_1, g_2,\dots, g_n$.
	Suppose that $h_1=\sigma_{\bfr_0}$, for some $\bfr_0\in\{0,1\}^{2n}$, is a symplectic partner of $g_1$ such that the commutation relations (\ref{eq:commutation}) hold. Then \[\ket{\phi^*}= h_1\ket{\phi},\]
	since they are both stabilized by  $-g_1, g_2,\dots, g_n$. Consequently, a Bell measurement on 
	$\ket{\phi}\otimes \ket{\phi}$ returns outcome $\bfr$
	with probability 
	\begin{align*}
		\left| \bra{\Phi_+}^{\otimes n} (I^{\otimes n}\otimes \sigma_\bfr)  \ket{\phi}\ket{\phi}   \right|^2
		=&\frac{1}{2^n}\left| \bra{\phi}\sigma_\bfr \ket{\phi^*}   \right|^2\\
		=&\frac{1}{2^n}\left| \bra{\phi}\sigma_{\bfr\oplus \bfr_0}\ket{\phi}   \right|^2. 
	\end{align*}
\end{proof}
\textbf{Remark:} The above corollary is slightly weaken than~\cite[Lemma 2]{Mon17},
where  the Pauli operator $\sigma_{\bfr_0}$ is shown to be $Z^{\bfs}$ for some $\bfs\in\{0,1\}^n$ 
by exploiting the  mathematical  structure  of a stabilizer state \cite{DDM03,VDN10}. 
This can be understood as $\ket{\phi}=   S^{\bfs}V\ket{0^n}$ for some $\bfs\in\{0,1\}^n$ and a unitary $V$ consisting of $H$, $X$, $CNOT$, and $CZ$ gates~\cite[Theorem 2]{VDN10}.

However, Corollary~\ref{cor:Bell_stab} is sufficient for our purpose of learning stabilizer states so that we have a self-contained proof here. Moreover, this idea will be exploited in learning $T$-depth one output states later.

\subsection{A Circuit Synthesis Method} \label{sec:circuit_synthesis}
In this section, we describe a circuit synthesis approach for Clifford circuits and it will be used later in Section~\ref{sec:learn_Clifford} for finding the circuit representation of an unknown target Clifford circuit.

An $n$-qubit state is described by a $2^n\times 2^n$ density matrix.
The $n$-fold Pauli operators are a basis for $2^n\times 2^n$ matrices.
To understand the evolution of a density operator under a unitary operator $U$, it suffices to know how $U$ operates on 
the $n$-fold Pauli operators. Namely, $U\rho U^\dag= \sum_\bfi \alpha_\bfi U\sigma_\bfi U^\dag$ with $\rho= \sum_\bfi \alpha_\bfi \sigma_\bfi$.
Since the Pauli matrices are related by $Y=iXZ$ (or $U Y_i U^\dag= i U X_i U^\dag\cdot U Z_i U^\dag$) and the $n$-fold Pauli group has $2n$ independent generators, one has the following lemma.

	\bl\label{lemma:2n} 
	Suppose  that $U$ is an   $n$-qubit  unitary operator.
	Given $UX_iU^\dag$ and $UY_jU^\dag$ (or $UZ_jU^\dag$) for $i,j=1,\dots,n$, $U$ can be uniquely determined up to a global phase.
	\el
	\begin{remark}
		Lemma~\ref{lemma:2n} shows that only one unitary $U$, up to a global phase, satisfies the pair constraints $(X_i, UX_i U^\dag)$ and $(Y_j, UY_jU^\dag)$ (or $(Z_j, UZ_jU^\dag)$) for $i,j=1,\dots,n$.
		However, it is computationally difficult to find the  circuit representation for $U$.
		{(For example, using na\"ive Gaussian elimination would take time $O(4^{3n})$.)}
	\end{remark}

In a learning task, the goal is usually to predict the output state $U\ket{\psi}$ or the measurement outcome on $U\ket{\psi}$.
Given knowledge of $UX_iU^\dag$ and $UY_jU^\dag$ for $i,j=1,\dots,n$, it is still not easy to do this prediction.
If we know the mathematical description of $\ket{\psi}$, we can find a Pauli decomposition of $\rho=\ket{\psi}\bra{\psi}= \sum_\bfi \beta_\bfi \sigma_\bfi$, and then evaluate the combination $U\rho U^\dag= \sum_\bfi \beta_\bfi U\sigma_\bfi U^\dag$. Since there are $4^n$ basis matrices, this decomposition is inefficient.
On the other hand, if  a polynomial-size circuit description of $U$ is available, we can simply apply $U$ to the input quantum state for predicting the output state.
In the following, we describe a crucial step---a circuit synthesis method---for finding a circuit representation for an unknown Clifford unitary.

\emph{Clifford circuits}  are composed of controlled-NOT (CNOT), Hadamard, and phase gates. 
A \textit{tableau} description of a Clifford unitary $U$ is a binary matrix with rows corresponding to $UX_iU^\dag$ and $UY_jU^\dag$ for $i,j=1,\dots,n$.
Given the tableau of $U$, 
Aaronson and Gottesman provided a circuit synthesis algorithm  that decomposes $U$ to a circuit that contains 11 stages of computation
in the sequence -H-C-P-C-P-C-H-P-C-P-C-~\cite{AG04},  where {-H-,} {-P-,} and -C- stand for stages composed of only  Hadamard, Phase, and CNOT gates, respectively.
(This is further improved to a nine-stage circuit by Maslov and Roetteler~\cite{MR18},
which can be utilized in fault-tolerant quantum computation~\cite{ZLB+20}.)
Consequently, any Clifford circuit can be decomposed into $O(n^2/\log n)$ Clifford gates with circuit depth $O(n)$~\cite{PMH08}  or  $O(n^2)$ Clifford gates with circuit depth $O(\log n)$~\cite{MN02}.
When the input to a Clifford circuit is restricted to $\ket{0^n}$, Van den Nest showed that the output state is equal to a Clifford circuit of five stages -H-C-X-P-CZ- on input $\ket{0^n}$~\cite{VDN10},
where -X- and -CZ- stand for stages of only $X$, and controlled phase gates, respectively.

In the following Lemma~\ref{lemma:Clifford_synthesis}, we show that if one only has access to an incomplete tableau, it is still possible to apply the above Clifford synthesis algorithm
with additional steps. This is similar to the encoding circuit decomposition for entanglement-assisted quantum stabilizer codes~\cite{KL19} and its proof  is omitted.

\bl \label{lemma:Clifford_synthesis}
Suppose that $U$ is a Clifford circuit. Given $UX_iU^\dag$ for $i=1,\dots,n$ and $UZ_jU^\dag$ for $j=1,\dots,t$ with $t\leq n$, we can construct a unitary operation $C$  composed of $O(n(n+t)/\log(n))$ Clifford gates such that $CX_iC^\dag=UX_iU^\dag$ for $i=1,\dots,n$ and $C Z_j C^\dag=UZ_jU^\dag$ for $j=1,\dots,t$. Moreover, such $C$ can be found in time $O(n^3)$.   \label{lemma:circuit_synthesis}
\el
Similarly, given $UX_iU^\dag, UY_jU^\dag$ for some $i,j$, one can derive a circuit $C$ such that $CX_iC^\dag=UX_iU^\dag, CY_jC^\dag=UY_jU^\dag$.

\subsection{Learning algorithm for unknown Clifford circuits} \label{sec:learn_Clifford}

In this section, we show how to learn a circuit representation for an unknown Clifford circuit $U\in \Cl_n$.
Our idea is similar to that of learning an unknown stabilizer state. For example, 
Montanaro proposed an algorithm for learning an unknown stabilizer state via identifying its stabilizer group~\cite{Mon17}.
However, identifying an unknown Clifford circuit is more complicated.
One would need to determine $UZ_iU^{\dag}$ and $UX_iU^{\dag}$ for $i=1,\dots, n$ simultaneously.  
Our key ingredient is to apply Lemma~\ref{lemma:circuit_synthesis} to find a Clifford circuit decomposition for $U$ as described in the following.

Our learning algorithm for unknown Clifford circuits is given in~Algorithm~\ref{Algorithm:I}.
We briefly explain how it works.
According to Corollary~\ref{cor:Bell_stab}, the set of Bell measurement outcomes $\{\bfr_i\oplus \bfr_0: i=1,\dots,2n\}$ obtained in step 2) is a set of stabilizers for $U\ket{0^n}$,
from which we can obtain a stabilizer group description of $U\ket{0^n}$.
Let  $\ket{\overline{0^n}}=U\ket{0^n}$ with stabilizers $  UZ_iU^{\dag}$ for $i=1,\dots,n$ and
$\ket{\overline{+^n}}=U\ket{+^n}$ with stabilizers $  UX_jU^{\dag}$ for $j=1,\dots,n$.
Using Motanaro's algorithm with $5n+2$ copies of $\ket{\overline{0^n}}$,
we can determine a  set of generators $g_1,\dots, g_n$, with  probability at least $(1-2^{-n})$,  such that 
\[
\langle g_1,\dots, g_n \rangle =\langle UZ_1U^{\dag},\dots,UZ_nU^{\dag} \rangle.
\]
Similarly, using $5n+2$ copies of $\ket{\overline{+^n}}$,
we can determine an independent set of generators $h_1,\dots, h_n$  with  probability at least $(1-2^{-n})$ such that 
\[
\langle h_1,\dots, h_n \rangle =\langle UX_1U^{\dag},\dots,UX_nU^{\dag} \rangle.
\]
Consequently, there exist $\textbf{b}_i\in\{0,1\}^n,$ for $i=1,\dots,n$, and
$\textbf{a}_j\in\{0,1\}^n,$ for $j=1,\dots,n$,  such that
\begin{align*}
	g_i =& UZ^{\textbf{b}_i}U^{\dag},\\
	h_j =& UX^{\textbf{a}_j}U^{\dag}.
\end{align*}
Once $\textbf{b}_i$ and $\textbf{a}_j$ are known, $UX_iU^\dag$ and $UZ_jU^\dag$ can be uniquely determined by Lemma~\ref{lemma:circuit_synthesis}.

\begin{algorithm}[h]

	\setcounter{AlgoLine}{0}
	\Input{  An oracle $\cO^{U}$, where $U\in\Cl_n$. }
	\Output{ A circuit description of $U$.}
	\begin{enumerate}[1)]
		\item  Prepare $5n+2$ copies of $\ket{\overline{0^n}}=U\ket{0^n}\in \mC^{2^n}$  using the oracle $\cO^{U}$.
		\item  For $i=0,\dots,2n$, perform a Bell measurement on  $\ket{\overline{0^n}}\otimes \ket{\overline{0^n}}$
		and denote the outcomes by $\bfr_i\in\{0,1\}^{2n}$.

		\item Determine a basis for $\{\sigma_{\bfr_i\oplus \bfr_0}: i=1,\dots,2n\}$ and denote the basis by  $\{g_1',\dots,g_n'\}$, where $g_i'\in\cP_n$.   
		
		\item  For $i=1,\dots, n$, do the following: 
		\begin{itemize}
			\item[-] 	Measure the Pauli operator $g_i'$ on $\ket{\overline{0^n}}$. 
			\item[-] If the outcome is $+1$, then $g_i=g_i'$ is a stabilizer of $\ket{\overline{0^n}}$; otherwise, $g_i=-g_i'$ is a stabilizer of $\ket{\overline{0^n}}$. 
		\end{itemize}
		
		Then $\{g_1,\dots,g_n\}$ is a stabilizer basis for $\ket{\overline{0^n}}$.

		\item  For $j=1,\dots, n$, do the following:
		\begin{itemize}
			\item[-]  prepare $n$ copies of $\ket{\overline{0^{j-1}10^{n-j}}}$ using the oracle $\cO^{U}$.
			\item[-]  	For $i=1,\dots,n$, measure $g_i$ on $\ket{\overline{0^{j-1}10^{n-j}}}$. 
			If the outcome is $+1$, set $b_{i,j}=0$; otherwise, $b_{i,j}=1$. 
		\end{itemize}

		\item   Find the inverse of $B=[b_{i,j}]$ and denote it by $B^{-1}=[d_{i,j}]$, using augmented matrices and Gaussian elimination.
		Then $$UZ_iU^\dag= \prod_{j=1}^n g_j^{d_{i,j}}$$ for $i=1,\dots, n$.	
		
		\item Repeat Steps 1) to 5) but with  $\ket{0},\ket{1}$, $g_j$, and $b_{i,j}$ replaced by $\ket{+},\ket{-}$, $h_j$ and $a_{i,j}$, respectively.
		\item Find the inverse of $A=[a_{i,j}]$ and denote it by $A^{-1}=[e_{i,j}]$, using augmented matrices and Gaussian elimination.
		Then $$UX_iU^\dag= \prod_{j=1}^n h_j^{e_{i,j}}$$ for $i=1,\dots, n$.	
		\item Apply Lemma~\ref{lemma:circuit_synthesis} to $\{UZ_iU^\dag, UX_jU^\dag: i,j=1,\dots, n\}$ and output the obtained Clifford circuit.
	\end{enumerate}

	\caption{Learning  an unknown Clifford circuit } \label{Algorithm:I}
\end{algorithm}

Finding a basis for  $\{\bfr_i\oplus \bfr_0: i=1,\dots,2n\}$ requires a Gaussian elimination, which takes time $O( n^3)$ in reality.  Similarly, Gaussian elimination is also needed in Lemma~\ref{lemma:circuit_synthesis}. 
The inverse of an matrix can be found by Gaussian elimination with augmented matrices. 
To sum up, we have the following theorem.
\bt \label{theorem:Clifford}
Given access to an oracle $\cO^U$, where $U$ is a Clifford unitary, one can identify $U$ using $2n^2+10n+4$ queries to $\cO^U$
in time $O(n^3)$ with probability at least $1-2^{-n+1}$.
\et

In the special case of learning simply an unknown stabilizer state $\ket{\psi}=U\ket{0^n}$,
applying steps 1) to 4) in Algorithm~\ref{Algorithm:I} is sufficient to find a set of stabilizer generators.

\section{Characterizing Quantum Circuits of Some $T$ gates} \label{sec:some_T}

\subsection{Expanded stabilizer formalism}

The Clifford gates
together with a non-Clifford gate, say $T=\begin{bmatrix}1&0\\0&e^{i\frac{\pi}{4}}\end{bmatrix}$, are universal for quantum computation \cite{NC00}. 
We will focus on quantum circuits composed of Clifford $+T$ gates in this paper. 
An arbitrary quantum circuit can be approximated by Clifford and $T$ gates and  this approximation can be decomposed as a sequence of Clifford stages and $T$ stages, alternatively. 
\bd
The number of $T$ stages in a circuit is called the \textit{$T$-depth} of the circuit.
The output state of a {$T$-depth} one circuit on input $\ket{0^n}$ will be called a \textit{\textit{$T$-depth} one output state}.
\ed 
For example, Figure~\ref{fig:UTU} provides a quantum circuit of $T$-depth one.
\begin{figure}[h]
	\[
	\includegraphics[width=8.cm]{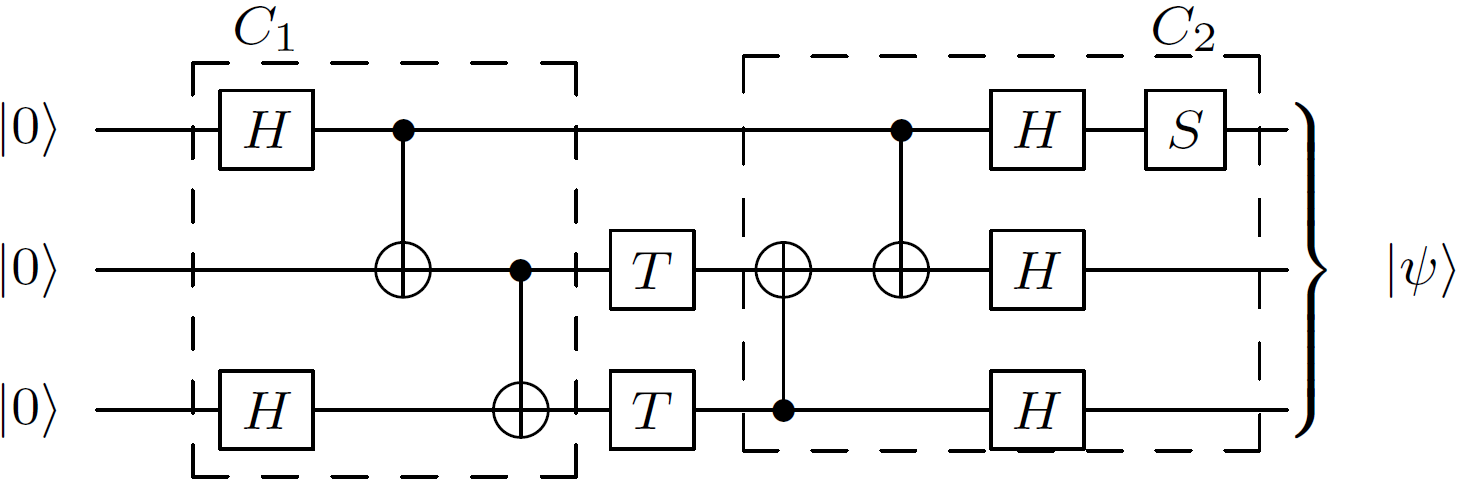}
	\]
	\caption{A quantum circuit of $T$-depth 1.} \label{fig:UTU}
\end{figure}

A $T$ gate can be implemented by the gadget shown in Figure~\ref{fig:Tgadget} with an ancillary magic state $\kT=\frac{\ket{0}+e^{i\frac{\pi}{4}}\ket{1}}{\sqrt{2}}$.
In addition, this $T$ gadget is, conditioned on the measurement outcome, equivalent to one of the postselected $T$ gadgets as shown in Figure~\ref{fig:postTgadgets}~\cite{BG16}, where only Clifford gates are required. Hence we have the following proposition.
\begin{figure}[h]
	\[
	\includegraphics[width=4.5cm]{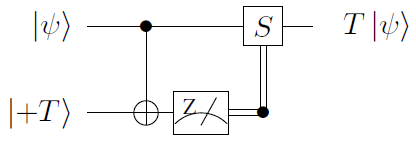}
	\]
	\caption{ $T$ gadget with a magic state $\kT$.} \label{fig:Tgadget}
\end{figure}
\begin{figure}[h]
	\[
	\includegraphics[width=5.5cm]{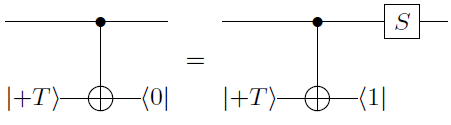}
	\]
	\caption{Postselected $T$ gadgets.} \label{fig:postTgadgets}
\end{figure}

\bp \label{prop:tgadget}
Any $n$-qubit Clifford+$T$ quantum circuit of $k$ $T$ gates
can be reduced to an $(n+k)$-qubit Clifford circuit of depth $O(\log(n+k))$ if postselection is possible and magic sates $\kT$ are available.
It can be further reduced to a quantum circuit of constant depth if, in addition, large ancillary state preparation is possible.
\ep
\begin{proof}

	Using the gadget in Figure~\ref{fig:postTgadgets}~\cite{BG16}, a quantum circuit of $k$ $T$ gates can be implemented by an equivalent Clifford circuit with $k$ ancillary magic states conditioned on outcome $0^k$ in the gadgets.
	(Both postselected $T$ gadgets work here, and  
	for simplicity, we use only the postselected gadget conditioned on outcome $0$ in the following.)
	Then the equivalent circuit has $n+k$ qubits and only Clifford gates, followed by some postselection measurements at the end.
	Since a Clifford circuit can be implemented with depth logarithmic in the number of qubits~\cite{MN02}, we have the first statement.
	
	As for the second statement, we simply use the gate teleportation technique~\cite{GC99,ZLB+20} to implement a Clifford circuit with a corresponding ancillary state.
	If this ancillary state can be prepared offline, the teleportation part can be done in constant depth.

\end{proof}

For example, the quantum circuit in Figure~\ref{fig:UTU} can
be implemented by Figure~\ref{fig:UTU2} with two postselected gadgets.
\begin{figure}[h]
	\[
	\includegraphics[width=8.cm]{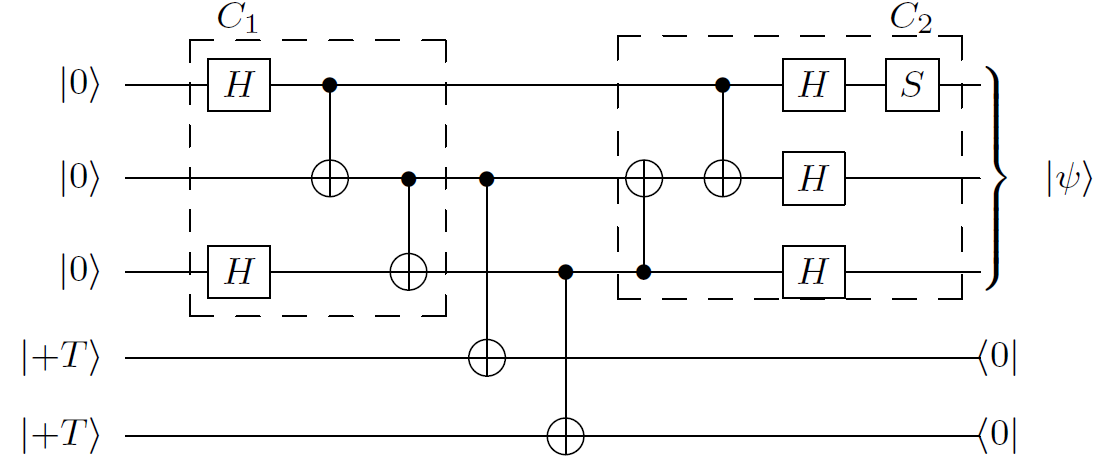}
	\]
	\caption{An  equivalent circuit of Figure~\ref{fig:UTU} with  postselected $T$ gadgets.} \label{fig:UTU2}
\end{figure}

Recently, it is shown that computations with larger quantum depth are strictly more powerful (with respect to an oracle)
and not every quantum algorithm can be implemented in logarithmic depth~\cite{CCL20,CM20}.  We remark that the above proposition does not violate the depth constraint since it assumes that postselection is available.

A quantum state $\rho$ can be represented as a  \textit{stabilizer pseudomixture} $\rho=\sum_\bfi \alpha_{\bfi} \sigma_\bfi$,
where $\sigma_\bfi$ are stabilizer states and $\alpha_{\bfi}$ are real numbers such that $\sum_\bfi \alpha_\bfi=1$.
Note that this representation is not unique and  $\alpha_{\bfi}$ can be negative.
The magic state $\kT\bT$ has the following two stabilizer pseudomixtures: 
\begin{align}
	\kT\bT=&\alpha_1 \ket{+}\bra{+} +\alpha_2 \ket{-}\bra{-}+\alpha_3\ket{+i}\bra{+i} \label{eq:T_decomp1}\\
	=&\alpha_1 \ket{+i}\bra{+i} +\alpha_2 \ket{-i}\bra{-i}+\alpha_3\ket{+}\bra{+}, \label{eq:T_decomp2}
\end{align}
where $\alpha_1=\frac{1}{2}$, $\alpha_2=\frac{1-\sqrt{2}}{2}$, and $\alpha_3=\frac{\sqrt{2}}{2}$.
In Eq.~(\ref{eq:T_decomp1}), the stabilizers of the component stabilizer states are $X$, $-X$, and $Y$, respectively, while
in Eq.~(\ref{eq:T_decomp2}), the stabilizers are  $Y$, $-Y$, and $X$, respectively.
In this work we will exploit the structure of stabilizer pseudomixtures in learning quantum circuits.
We will provide algorithms to learn Clifford $+T$ circuits of 
some $T$ gates,   specifically $O(2^{{\log n}})$.

We  propose the   following  \emph{expanded stabilizer formalism} for simulating a quantum circuit $U$ of 
$k$ $T$ gates
when the input is a stabilizer state. First we replace the $k$ $T$ gates in the target quantum circuit $U$ by  postselected $T$ gadgets,
which leads to a Clifford circuit with input a pseudomixture of $3^k$ stabilizer states.
Then the evolution of the input stabilizer state under $U$ can be done by 
tracing the $3^k$ corresponding Pauli frames in the remaining Clifford circuit
and then combining the resulting states appropriately.
We say that these $3^k$ Pauli frames constitutes an \emph{expanded Pauli frame}.
Clearly, this method can efficiently handle $k=O(\log n)$ $T$ gates.
\be Consider a $T$ gate operating on $\ket{+}$, which is stabilized by $X$. The output is $\kT= T\ket{+}$ and the expanded Pauli frame evolves from $X$ to $(X,-X,Y)$  as follows:
\begin{align*}
	&\begin{pmatrix}
		+00&10\\
		+00&01
	\end{pmatrix}, \begin{pmatrix}
		+00&10\\
		-00&01
	\end{pmatrix}, \begin{pmatrix}
		+00&10\\
		+01&01
	\end{pmatrix}\\
	\stackrel{\mbox{CNOT$_{1,2}$}}{\longrightarrow}
	&\begin{pmatrix}
		+00&11\\
		+00&01
	\end{pmatrix}, \begin{pmatrix}
		+00&11\\
		-00&01
	\end{pmatrix}, \begin{pmatrix}
		+00&11\\
		+11&01
	\end{pmatrix}\\
	=&\begin{pmatrix}
		+00&10\\
		+00&01
	\end{pmatrix}, \begin{pmatrix}
		-00&10\\
		-00&01
	\end{pmatrix}, \begin{pmatrix}
		+11&10\\
		+11&01
	\end{pmatrix}\\
	\stackrel{\mbox{$I\otimes\bra{0}_2$}}{\longrightarrow}
	&\begin{pmatrix}
		+00&10\\
		+01&00
	\end{pmatrix}, \begin{pmatrix}
		-00&10\\
		+01&00
	\end{pmatrix}, \begin{pmatrix}
		+11&10\\
		+01&00
	\end{pmatrix}\\
	\stackrel{\mbox{ discard  the ancilla}}{\longrightarrow}
	&\begin{pmatrix}
		+0&1\\
	\end{pmatrix}, \begin{pmatrix}
		-0&1\\
	\end{pmatrix}, \begin{pmatrix}
		+1&1\\
	\end{pmatrix}.
\end{align*}
\eep

In general,  we can start with an $n$-qubit Pauli frame. Each time  when a $T$ gate is applied,
we consider the $n+1$ qubit Pauli frames, expand the number of Pauli frames by three, do CNOTs, measure the ancilla with postselected outcome $0$, and then discard the ancilla qubit. 
Continuing this process, we end with $3^k$ $n$-qubit Pauli frames.

\textbf{Remark:} in the task of classical simulation in~\cite[Theorem 3]{BSS16}, knowledge of  Eq.~(\ref{eq:T_decomp1})
is sufficient  so that classical measurement outcomes can be linearly combined.
However, we need   both  Eqs.~(\ref{eq:T_decomp1}) and~(\ref{eq:T_decomp2}) to develop our expanded stabilizer formalism as shown in the following.

The transformations of single Pauli operators under $T$ in the expanded Pauli frame are summarized in Table~\ref{table:expandedT}. 
A triplet such as $(X,-X,Y)$ means that the state is a pseudomixture of the three states stabilized by $X$, $-X$, and $Y$, and with coefficients $\alpha_1,\alpha_2,$ and $\alpha_3,$ respectively.
Note that $Z$ commutes with $T$ so $\pm Z$ remains unchanged.
The two triplets in an entry correspond to the expressions~(\ref{eq:T_decomp1}) and~(\ref{eq:T_decomp2}), respectively. 
In each triplet, the second operator is always equal to the first, multiplied by $-1$.
Observe that these two triplets have the same first and third operators in the opposite order. 
Therefore, we will call the first and third operators as the \textit{primary symplectic stabilizers} (they are also symplectic partners to each other).
For example, we say that $X$ and $Y$ are the {\pss} of $\kT$.
Another reason is that when measuring each of these {\pss} on the output state,
we obtain outcome $+1$ with probability $\frac{2+\sqrt{2}}{4}$
and outcome $-1$ with probability $\frac{2-\sqrt{2}}{4}$.
These ideas naturally  extend to $n$-fold Pauli operators.

\begin{table} [h]
	%	\begin{center}
		\resizebox{\columnwidth}{!}{%
			\begin{tabular}{|c|cc|c|}
				\hline
				\hline
				input & output&& {\pss} \\
				\hline
				$X$ & ($X$,$-X$, $Y$)& ($Y$,$-Y$, $X$)& $X, Y$\\
				\hline
				$-X$ & ($-X$,$X$, $-Y$)& ($-Y$,$Y$, $-X$)& $-X,-Y$ \\
				\hline
				$Y$ & ($Y$,$-Y$, $-X$) 	  & ($-X$,$X$, $Y$)& $-X,Y$\\
				\hline
				$-Y$ & ($-Y$,$Y$, $X$)		& ($X$,$-X$, $-Y$)& $X,-Y$\\
				\hline
				$\pm Z$ & $\pm Z$& & NA\\
				\hline	
				\hline
		\end{tabular}}
		%	\end{center}
	\caption{Expanded stabilizer formalism of single Pauli operators under the operation $T$.
	}
	\label{table:expandedT}
\end{table}
\subsection{Quantum circuits of $T$-depth one}\label{sec:learning}

An expanded Pauli frame of an $n$-qubit quantum circuit  has exponentially many Pauli frames in the number of $T$ gates and this seems intractable when $k$ goes large.  For $T$-depth one quantum circuits, we show   that it suffices to trace the evolution of at most $2n$ {\pss}.
Without loss of generality, we consider a $T$-depth one quantum circuit $U= C_2  T^\bfv   C_1$,
where $C_1$ and $C_2$ are Clifford operators in $\Cl_n$, and $\bfv\in\{0,1\}^n$ of Hamming weight $k$. 
Here $T^\bfv$ applies a $T$ gate to qubit $i$ if $\bfv_i=1$, and operates trivially, otherwise. 
Figure~\ref{fig:UTU} illustrates such an example.

 {
	\bd
	Suppose that $U= C_2  T^\bfv   C_1$ is a $T$-depth one circuit.
	$U$ is said to be of \textit{full $T$-rank }if  $C_1\ket{0}$ restricted on the support of the $T$ gates has full $X$-rank
	and $U\ket{0}$ is called a \textit{full-rank} $T$-depth one output state.
	\label{def:trank}
	\ed}

\bl Suppose that $U= C_2  T^\bfv   C_1$ is a $T$-depth one circuit, where $C_1,C_2\in\Cl_n$ and $\bfv\in\{0,1\}^n$ is of  Hamming weight $k\leq n$,  {and $U$ is of full $T$-rank.}
% {which is no larger than the $X$-rank  of $C_1\ket{0}$}.
Then  $U\ket{0^n}$ has a pseudomixture of   $3^{\hat{k}}$ orthogonal stabilizer states, where $\hat{k}\leq k$,   
and its expanded Pauli frame has $2\hat{k}$  \pss\ and the other $n-\hat{k}$ stabilizer generators that stabilize all the component stabilizer states.
\label{lemma:3tok}
\el
\begin{proof}
	
	Recall the definition of $X$-rank in Def.~\ref{def:xrank}. 
	In the following we consider the two cases of whether $C_1\ket{0^n}$ has a Pauli frame of 
	full $X$-rank.
	
\noindent  
  Assume that $C_1\ket{0^n}$ has a Pauli frame of
		full $X$-rank.
		Without  loss of generality,
		we may assume that the Pauli frame of $C_1\ket{0^n}$ has the following form:
		\begin{align*}
			\left(
			\begin{array}{ccccc|cccc}
				\sgn{c_1}& w_{11}&w_{12}&\cdots &w_{1n}&1&0&\cdots&0\\
				\sgn{c_2}& w_{21}&w_{22}&\cdots &w_{2n}&0&1&\cdots&0\\
				\vdots&\vdots&\vdots &\ddots&\vdots&\vdots&\ddots&\vdots\\
				\sgn{c_n}& w_{n1}&w_{n2}&\cdots &w_{nn}&0&0&\cdots&1\\
			\end{array}
			\right),
		\end{align*}
		where $\bfc\in\{+1,-1\}^{n\times 1}$ and $ w_{ij} \in\{0,1\} $.
		This can be done by   appropriate row multiplications.
		
		Assume $w_{nn}=0$. The case that $w_{nn}=1$ is similar.
		Now we simulate the evolution of a postselected gadget (conditioned on outcome $0$) operating on the last qubit of $U_1 \ket{0^n}$, that is, applying a $T$ gate to the last qubit.
		By Eq.~(\ref{eq:T_decomp1}), we start with three Pauli frames (of dimension $(n+1)\times (2(n+1)+1) $):
		\begin{align*}
			&\left(
			\begin{array}{ccccc|ccccc}
				\sgn{c_1}& w_{11}&\cdots &w_{1n}&0&1& \cdots&0&0\\
				\vdots&\vdots&\ddots&\vdots&\vdots&\vdots&\ddots&\vdots&\vdots\\
				\sgn{c_n}& w_{n1}&\cdots &0&0&0&\cdots&1&0\\
				\pm &0&\cdots &0&0&0&\cdots&0&1
			\end{array}
			\right),\\
			&\left(
			\begin{array}{ccccc|ccccc}
				\sgn{c_1}& w_{11}&\cdots &w_{1n}&0&1& \cdots&0&0\\
				\vdots&\vdots&\ddots&\vdots&\vdots&\vdots&\ddots&\vdots&\vdots\\
				\sgn{c_n}& w_{n1}&\cdots &0&0&0&\cdots&1&0\\
				&0&\cdots &0&1&0&\cdots&0&1
			\end{array}
			\right)
		\end{align*}
		\begin{align*}
			\stackrel{\text{CNOT}}{\longrightarrow} 
			&\left(
			\begin{array}{ccccc|ccccc}
				\sgn{c_1}& w_{11}&\cdots &w_{1n}&0&1& \cdots&0&0\\
				\vdots&\vdots&\ddots&\vdots&\vdots&\vdots&\ddots&\vdots&\vdots\\
				\sgn{c_n}& w_{n1}&\cdots &0&0&0&\cdots&1&1\\
				\pm &0&\cdots &0&0&0&\cdots&0&1
			\end{array}
			\right),
			\end{align*}
	\begin{align*}	&\left(
			\begin{array}{ccccc|ccccc}
				\sgn{c_1}& w_{11}&\cdots &w_{1n}&0&1& \cdots&0&0\\
				\vdots&\vdots&\ddots&\vdots&\vdots&\vdots&\ddots&\vdots&\vdots\\
				\sgn{c_n}& w_{n1}&\cdots &0&0&0&\cdots&1&1\\
				&0&\cdots &1&1&0&\cdots&0&1
			\end{array}
			\right)
		\end{align*}
	\begin{align*}=&\left(
			\begin{array}{ccccc|ccccc}
				\sgn{c_1}& w_{11}&\cdots &w_{1n}&0&1& \cdots&0&0\\
				\vdots&\vdots&\ddots&\vdots&\vdots&\vdots&\ddots&\vdots&\vdots\\
				\sgn{\pm c_n}& w_{n1}&\cdots &0&0&0&\cdots&1&0\\
				\pm &0&\cdots &0&0&0&\cdots&0&1
			\end{array}
			\right),\\
			&\left(
			\begin{array}{ccccc|ccccc}
				\sgn{c_1}& w_{11}&\cdots &w_{1n}&0&1& \cdots&0&0\\
				\vdots&\vdots&\ddots&\vdots&\vdots&\vdots&\ddots&\vdots&\vdots\\
				\sgn{c_n}& w_{n1}&\cdots &1&1&0&\cdots&1&0\\
				&0&\cdots &1&1&0&\cdots&0&1
			\end{array}
			\right)
		\end{align*}
		\begin{align*}
			\stackrel{\bra{0}}{\longrightarrow} 
			&\left(
			\begin{array}{ccccc|ccccc}
				\sgn{c_1}& w_{11}&\cdots &w_{1n}&0&1& \cdots&0&0\\
				\vdots&\vdots&\ddots&\vdots&\vdots&\vdots&\ddots&\vdots&\vdots\\
				\sgn{\pm c_n}& w_{n1}&\cdots &0&0&0&\cdots&1&0\\
				&0&\cdots &0&1&0&\cdots&0&0
			\end{array}
			\right),\\
			&\left(
			\begin{array}{ccccc|ccccc}
				\sgn{c_1}& w_{11}&\cdots &w_{1n}&0&1& \cdots&0&0\\
				\vdots&\vdots&\ddots&\vdots&\vdots&\vdots&\ddots&\vdots&\vdots\\
				\sgn{c_n}& w_{n1}&\cdots &1&1&0&\cdots&1&0\\
				&0&\cdots &0&1&0&\cdots&0&0
			\end{array}
			\right).
		\end{align*}
		Since the ancilla stays in $\ket{0}$ and can be discarded, the three $n$-qubit Pauli frames become
		\begin{align}
			&\left(
			\begin{array}{cccc|cccc}
				\sgn{c_1}& w_{11}&\cdots &w_{1n}&1& \cdots&0\\
				\vdots&\vdots&\ddots&\vdots&\vdots&\ddots&\vdots\\
				\sgn{\pm c_n}& w_{n1}&\cdots &0&0&\cdots&1
			\end{array}
			\right),\\
			&\left(
			\begin{array}{cccc|cccc}
				\sgn{c_1}& w_{11}&\cdots &w_{1n}&1& \cdots&0\\
				\vdots&\vdots&\ddots&\vdots&\vdots&\ddots&\vdots\\
				\sgn{c_n}& w_{n1}&\cdots &1&0&\cdots&1
			\end{array}
			\right). \label{eq:3PF}
		\end{align}
		Note that the first two Pauli frames only differ by a sign on the last row.
		Each of these three frames has an identity on the right half.

		Repeating this process on different qubits, we end up with $3^{k}$ Pauli frames.
		Clearly these Pauli frames are different and they correspond to orthogonal stabilizer states.
		After each $T$ gate, we have one additional \ps. 
		Therefore, these $3^k$ Pauli frames share $n+k$ independent generators up to a phase $-1$
		or the expanded Pauli frame of $T^{\bfv}C_1\ket{0^n}$ has $2k$ \pss\ and the other $n-k$ stabilizer generators that stabilize all the components stabilizer states.  
		Finally these Pauli frames evolve according to the Clifford unitary $C_2$, which preserves the expanded Pauli frame structure.

 {	
If $C_1\ket{0^n}$ is not of full $X$-rank,
we can do row operations on the Pauli frame so that  its restriction on the support of the $T$ gates is of full $X$-rank by assumption.
Then the statement holds following a similar argument.}

\end{proof}

\textbf{Remark:}
The above lemma says that the output state $\ket{\psi}$ of a $T$-depth one quantum circuit   {of full $T$-rank} on input a computational basis vector can be represented by a pseudomixture of $3^k$ stabilizer states ($k\leq n$), which share
$n+k$ independent Pauli generators up to a phase $-1$.
Among the $n+k$ independent Pauli generators, $2k$ of them  are  symplectic partners, say $g_{n-k+1},h_{n-k+1},\dots,g_n,h_n$, and are   primary symplectic stabilizers.  The remaining $n-k$ operators, say $g_{1},\dots,g_{n-k}$, stabilize all the component stabilizer states and will be called 
\emph{isotropic stabilizers}.
Note that  two \pss\ are equivalent up to  a product of the isotropic stabilizers. The following is the set of all the stabilizers of the $3^k$ stabilizer states up to a phase $-1$:
\begin{align*}
&\langle g_1,\dots,g_{n-k} \rangle \times \{I^{\otimes n},g_{n-k+1},h_{n-k+1}\}\\
&\times  \{I^{\otimes n},g_{n-k+2},h_{n-k+2}\} \times \cdots \times  \{I^{\otimes n},g_{n},h_{n}\}.
\end{align*}
Here the set product is defined by 
\begin{align*}
	&\{x_1, \dots,x_m \}\times \{y_1,\dots,y_l \}\\
	=&\{x_1y_1,x_1y_2,\dots,x_1y_l,x_2y_1,\dots,x_my_l \}.  
\end{align*}
Without loss of generality, the $T$-depth one output state $\ket{\psi}$ admits the following stabilizer pseudomixture
\begin{align}
	\ket{\psi}\bra{\psi}= \sum_{\bfj\in\{1,2,3\}^k } \beta_\bfj   \ket{\phi_{\bfj}}\bra{\phi_\bfj}, \label{eq:stabilizer_pseduomixture}
\end{align}
where $\beta_\bfj= \prod_{i=1}^k \alpha_{\bfj_i}$ and $\phi_{\bfj}$ is stabilized by $g_{n-k+i}$ if $\bfj_i=1$,  $-g_{n-k+i}$ if $\bfj_i=2$, or $h_{n-k+i}$ if $\bfj_i=3$, for $i=1,\dots, k$.

\be  $T\otimes   I (\ket{00}+\ket{11})/\sqrt{2}$ has an expanded Pauli frame generated by two \pss\ $\{ XX,XY\}$ and an isotropic stabilizer $ZZ$.
\eep

\be If $C_1\ket{0^n}$ is a graph state, then its Pauli frame has full $X$-rank.
\eep

\be  $T\otimes T\otimes I (\ket{000}+\ket{111})/\sqrt{2}=S\otimes I\otimes I (\ket{000}+\ket{111})/\sqrt{2}$, which is a stabilizer state.
\eep

\be The single-qubit state $THT\ket{+}$ is not a $T$-depth one output state, and it has  a pseudomixture of five stabilizer states:
$\alpha_1 \ket{0}\bra{0}+\alpha_2 \ket{1}\bra{1}+ \alpha_3\alpha_1 \ket{+i}\bra{+i}+\alpha_3\alpha_2\ket{-i}\bra{-i}+\alpha_3^2\ket{-}\bra{-}$.
\eep

For classical simulation of such a quantum circuit, it can be done by 
simulating each stabilizer state and linearly combining the $3^{k}$ outputs~\cite{BSS16}.
The central idea of our learning method is to reverse the above simulation process. 
	However, the task of identifying states is more difficult since we do not have the $3^{k}$ Pauli frames to start with, but only with several copies of the final state.
In the following section we will provide an algorithm to derive an effective quantum circuit that maps the  all-zero state to the state corresponding to the unknown $3^{k}$ Pauli frames for $k=O( {\log(n)})$.

\section{Learning quantum circuits of certain magic} \label{sec:learning_T}

\subsection{Efficiently Learning $T$-depth one output states with input computational basis states}

In this section we study learning of the output of an unknown  quantum circuit $U$ of {some} $T$ gates, given oracle access to $\cO^{U}$, on input a computational basis state $\ket{0^n}$.
We make an additional assumption that the conjugate oracle $\cO^{U^*}$ is also available
since this will simplify the learning task.
	In the next section, we will show that this assumption can be removed at an additional cost.

Suppose that $U$ is a $T$-depth one quantum circuit  {of full $T$-rank}.	Let  $\ket{\psi}= U\ket{0^n}$ be the target quantum state.
Lemma~\ref{lemma:3tok} indicates that the state $\ket{\psi}=U\ket{0^n}$ is a pseudo mixture of at most $3^k$ stabilizer states,
which have at most $2n$ \pss. 
Consequently, one can identify $\ket{\psi}$  
by finding its \pss.

Lemma~\ref{lemma:BellMeasure} says that one can learn a stabilizer, up to a phase, of a stabilizer state by a Bell measurement.  In the following Lemma~\ref{lemma:T1_state},
we analyze the outcomes of Bell measurements on $T$-depth one output states.
\bl\label{lemma:T1_state}
Suppose $\ket{\psi}\bra{\psi}= \sum_{\bfj\in\{1,2,3\}^k} \beta_\bfj \ket{\phi_\bfj}\bra{\phi_\bfj}$ is a  {full-rank} $T$-depth one  output state as defined in~Eq.\,(\ref{eq:stabilizer_pseduomixture}).
Let $\bfr\in\{0,1\}^{2n}$ be the outcome of a joint Bell measurement on $\ket{\psi^*}\otimes \ket{\psi}$.
\begin{enumerate}
	\item If $(\pm) \sigma_\bfr$ is not a stabilizer of any of the $\ket{\phi_\bfj}\bra{\phi_\bfj}$,
	the Bell measurement returns $ \bfr$ with zero probability.
	
	\item If $(\pm) \sigma_\bfr$ is an isotropic stabilizer, the Bell measurement returns outcome $\bfr$   with probability $1/2^n$.
	\item If $ (\pm)  \sigma_\bfr$  is a {\ps}, the Bell measurement returns outcome $\bfr$   with probability $\frac{1}{2^{n+1}}$.

	\item  if $ (\pm)  \sigma_\bfr$  is a product of $m$ primary symplectic stabilizers,  the Bell measurement returns outcome $\bfr$   with probability $\frac{1}{2^{n+m}}$.

\end{enumerate}
\label{lemma:BellMeasure2}	
\el
\begin{proof}
	
	The outcome of a joint Bell measurement on  $\ket{\psi^*}\otimes \ket{\psi}$ is $\bfr\in\{0,1\}^{2n}$
	with probability 
	\begin{align*}
		\left| \bra{\Phi_+}^{\otimes n} (I\otimes \sigma_\bfr)  \ket{\psi^*}\ket{\psi}   \right|^2
		=&\frac{1}{2^n}\left|\sum_{\bfj} \beta_\bfj \bra{\phi_\bfj}\sigma_\bfr \ket{\phi_\bfj}   \right|^2. 
	\end{align*}
	
	\begin{enumerate}
		\item If $(\pm) \sigma_\bfr$ is not a stabilizer of any  $\ket{\phi_\bfj}\bra{\phi_\bfj}$,
		then for each $\ket{\phi_\bfj}$, there exists a stabilizer $g\in\cP_n$ of $\ket{\phi_\bfj}$ such that $g\sigma_\bfr=-\sigma_\bfr g$ and thus $\bra{\phi_\bfj}\sigma_\bfr\ket{\phi_\bfj}=\bra{\phi_\bfj}\sigma_\bfr 
		\left(\frac{I^{\otimes n}+g}{2}\right)\ket{\phi_\bfj}=  \bra{\phi_\bfj}\left(\frac{I^{\otimes n}-g}{2}\right)\sigma_\bfr\ket{\phi_\bfj}=0$.
		Thus, the Bell measurement returns $ \bfr$ with zero probability.
		
		\item If $(\pm) \sigma_\bfr$ is an isotropic stabilizer,  $\bra{\phi_\bfj}\sigma_\bfr\ket{\phi_\bfj}=1$ for all $\bfj$ and then the Bell measurement returns outcome $\bfr$   with probability $1/2^n$.
		\item If $ (\pm)  \sigma_\bfr$  is a {\ps}, say $h_{n-k+1}$, then it will stabilize exactly $1/3$ of the states; $\sigma_\bfr$'s symplectic partner, say $g_{n-k+1}$,  will stabilize another $1/3$;
		and $-g_{n-k+1}$ will stabilize the other $1/3$. Only the terms   stabilized by $\sigma_\bfr$ have contribution to the probability and they sum to $\alpha_3$. 
		Thus
		\begin{align*}
			&\frac{1}{2^n}\left|\sum_{\bfj\in\{1,2,3\}^k} \beta_\bfj \bra{\phi_\bfj}\sigma_\bfr \ket{\phi_\bfj}   \right|^2\\
			=&\frac{1}{2^n}\left|\sum_{\bfj\in \{1,2,3\}^k, \bfj_1=3} \beta_\bfj \bra{\phi_\bfj}\sigma_\bfr \ket{\phi_\bfj}   \right|^2\\
			=&\frac{1}{2^n} \alpha_3^2\\
			=&\frac{1}{2^{n+1}}>0.
		\end{align*}

		\item Similarly, if $ (\pm)  \sigma_\bfr$  is a product of $m$ primary symplectic stabilizers, say $h_{n-k+1}h_{n-k+2}\cdots h_{n-k+m}$,
		then  	\begin{align*}
			&\frac{1}{2^n}\left|\sum_{\bfj\in\{1,2,3\}^k} \beta_\bfj \bra{\phi_\bfj}\sigma_\bfr \ket{\phi_\bfj}   \right|^2\\
			=&\frac{1}{2^n}\left|\sum_{\bfj\in \{1,2,3\}^k, \bfj_1=\cdots=\bfj_m=3} \beta_\bfj \bra{\phi_\bfj}\sigma_\bfr \ket{\phi_\bfj}   \right|^2\\
			=&\frac{1}{2^n} \alpha_3^{2m}\\
			=&\frac{1}{2^{n+m}}>0.
		\end{align*}
	\end{enumerate}
\end{proof}

\textbf{Remark:}  the Bell measurement returns  a \ps,  up to a phase, with higher probability than a product of some \pss.

In the previous discussion, we assumed the availability of copies of the conjugate state $\ket{\psi^*}$ of a $T$-depth one output state $\ket{\psi}$ of $U$ on input $\ket{0^n}$.  If copies of the conjugate state $\ket{\psi^*}$ are not available, it can be very expensive to construct $\ket{\psi^*}$ using $\cO^{U}$~\cite{MSM19}. 
In the following, we prove a lemma for Bell measurement on two copies of a  {full-rank} $T$-depth one output state, which is similar to Corollary~\ref{cor:Bell_stab} for stabilizer states.
\bl \label{lemma:BellMeasure3}
Suppose that $\ket{\psi}\bra{\psi}= \sum_{\bfj\in\{1,2,3\}^k} \beta_\bfj \ket{\phi_\bfj}\bra{\phi_\bfj}$ is a  {full-rank} $T$-depth one  output state as defined in~Eq.\,(\ref{eq:stabilizer_pseduomixture}). 
Then there exists $\bfr_0\in\{0,1\}^{2n}$ such that
a joint Bell measurement on  $\ket{\psi}\otimes \ket{\psi}$ returns outcome $ \bfr$
with probability 
\[
\frac{\left|\bra{\psi}\sigma_{\bfr}\ket{\psi^*}\right|^2}{2^n}=\frac{\left|\bra{\psi}\sigma_{\bfr\oplus \bfr_0}\ket{\psi}\right|^2}{2^n}.
\]
\el
\begin{proof}
	
	Suppose that $\ket{\psi}$ has an expanded Pauli frame generated by $g_1,\dots,g_{n},h_{n-k+1},\dots,h_n\in\cP_n$. 
	Thus
	$\ket{\psi^*}$ has an expanded Pauli frame generated by $g_1^*,\dots,g_{n}^*,h_{n-k+1}^*,\dots,h_n^*$,
	where $g_i^* = g_i$ if the number of its Pauli component $Y$  is even, and $g_i^* = -g_i$, otherwise. Similarly, $h_i^*=\pm h_i$.
	
	We consider two cases that whether there exists an isotropic stabilizer  that is equal to its conjugate. 
	First, suppose that one isotropic stabilizer $g_1$ is such that $g_1^* = -g_1$. We may assume that
	$g_i^* = g_i$ for $i=2,\dots, n$ and $h_i^*=  h_i$ for $i=n-k+1,\dots, n$.
	Therefore, $\ket{\phi_\bfj^*}$ is stabilized by operators in $\langle -g_1,g_2,\dots,g_{n-k} \rangle\times \{I^{\otimes n}, g_{n-k+1},h_{n-k+1}\}\times \{I^{\otimes n},g_{n-k+2},h_{n-k+2}\}\times \cdots \times\{I^{\otimes n},g_n,h_n\}$. 
	Consider   $\ket{\phi_{11\cdots 1}^*}$, which is stabilized by
	$-g_1,g_2,\dots,g_n$. 	Let $h_1=\sigma_{\bfr_0}$ for some $\bfr_0\in\{0,1\}^{2n}$ be a symplectic partner of $g_1$ such that the commutation relations (\ref{eq:commutation}) hold.  Then we have 
	\[
	\ket{\phi_{11\cdots 1}^*}= h_1 \ket{\phi_{11\cdots 1}}
	\]
	since $h_1 \ket{\phi_{11\cdots 1}}$ is also stabilized by $-g_1,g_2,\dots,g_n$. Similarly, we can show that 
	\[
	\ket{\phi_{\bfj}^*}\bra{\phi_{\bfj}^*}= h_1 \ket{\phi_{\bfj}}\bra{\phi_{\bfj}} h_1^\dag   \mbox{ for all } \bfj\in\{1,2,3\}^k.
	\]
	This implies that 
	\[
	\ket{\psi^*}\bra{\psi^*}= h_1\ket{\psi}\bra{\psi}h_1^\dag. 
	\]
	Consequently, a joint Bell measurement on  $\ket{\psi}\otimes \ket{\psi}$ returns outcome  $\bfr\in\{0,1\}^{2n}$
	with probability 
	\begin{align*}
		\left| \bra{\Phi_+}^{\otimes n} (I\otimes \sigma_\bfr)  \ket{\psi}\ket{\psi}   \right|^2 
		=&\frac{1}{2^n}{ \bra{\psi} \sigma_\bfr \ket{\psi^*}  \bra{\psi^*} \sigma_\bfr^\dag  \ket{\psi}  } \\
		=&\frac{1}{2^n}\tr{\left( \ket{\psi} \bra{\psi} \sigma_\bfr \ket{\psi^*}  \bra{\psi^*} \sigma_\bfr \right) } \\
		=&\frac{1}{2^n}\tr{\left( \ket{\psi} \bra{\psi} \sigma_\bfr h_1\ket{\psi}  \bra{\psi} h_1^\dag  \sigma_\bfr \right) } \\
		=& \frac{\left|\bra{\psi}\sigma_{\bfr\oplus \bfr_0}\ket{\psi}\right|^2}{2^n}.
	\end{align*}
	
	Now suppose that 	$g_i^* = g_i$ for $i=1,\dots, n-k$. Let $ \cG=\{i:\ n-k+1\leq i\leq n,\ g_i^*=-g_i\}$ and $\cH=\{i:\ n-k+1\leq i\leq n,\ h_i^*=- h_i\}.$
	Let 
	\begin{align}
		\sigma_{\bfr_0}= \prod_{i\in\cG} h_i \prod_{l\in\cH}g_l. \label{eq:sigma_r0}
	\end{align}  
	Consider $\ket{\phi_\bfj^*}$ with $\bfj= \bfj_1 \bfj_2 \cdots \bfj_k$. We have
	\begin{align*}
		 \ket{\phi_\bfj^*}\bra{\phi_\bfj^*} 
		\stackrel{(a)}{=}& \left(\prod_{i\in\cG,\ \bfj_i=1,2} h_i \prod_{l\in\cH,\ \bfj_l=3}g_l \right) \ket{\phi_\bfj} \bra{\phi_\bfj}\\
		&\cdot 
		\left( \prod_{i\in\cG,\ \bfj_i=1,2} h_i \prod_{l\in\cH,\ \bfj_l=3}g_l \right)^\dag \\
		\stackrel{(b)}{=}& \sigma_{\bfr_0} \ket{\phi_{\bfj}}\bra{\phi_{\bfj}} \sigma_{\bfr_0}^\dag   \mbox{ for all } \bfj\in\{1,2,3\}^k.
	\end{align*}    
	Equality $(a)$ is because that the Pauli operator  $\prod_{i\in\cG,\ \bfj_i=1,2} h_i \prod_{l\in\cH,\ \bfj_l=3} g_l$ is used to modify the phases of the stabilizers of $ \ket{\phi_\bfj}$ so that
	the modified state has the same stabilizers as  $\ket{\phi_\bfj^*}$.
	Next, for $i\in\cG$ and $\bfj_i=3$, $h_i$ is a stabilizer of $\ket{\phi_{\bfj}}$
	and similarly, for $l\in\cH$ and $\bfj_l=1,2$, $\pm g_l$ is a stabilizer of $\ket{\phi_{\bfj}}$. Hence  equality $(b)$ holds. 
	Following the same steps as above, we can show that a joint Bell measurement on  $\ket{\psi}\otimes \ket{\psi}$ returns outcome  $\bfr\in\{0,1\}^{2n}$
	with probability $\frac{\left|\bra{\psi}\sigma_{\bfr\oplus \bfr_0}\ket{\psi}\right|^2}{2^n}.$	
\end{proof}

	Lemma~\ref{lemma:BellMeasure} and Lemma~\ref{lemma:BellMeasure3} suggest that a Bell measurement on $\ket{\psi}\otimes \ket{\psi}$ returns an outcome that corresponds to a stabilizer of its component states, up to a phase, times a fixed Pauli operator.
	Therefore, we can find a set of generators $g_1',g_2',\dots,g_n',h_{n-k+1}',h_n',\dots,h_k'$ by Bell sampling on  $\ket{\psi}\otimes \ket{\psi^*}$ for sufficiently many times according to Lemma~\ref{lemma:BellMeasure2} and Lemma~\ref{lemma:BellMeasure3} so that the $2k$ target \pss\ can be generated from  $g_1',g_2',\dots,g_n',h_{n-k+1}',\dots,h_n'$.

\bl
Suppose that $\ket{\psi}\bra{\psi}= \sum_{\bfj\in\{1,2,3\}^k} \beta_\bfj \ket{\phi_\bfj}\bra{\phi_\bfj}$ is a  {full-rank} $T$-depth one  output state as defined in~Eq.\,(\ref{eq:stabilizer_pseduomixture}). 
Performing a Bell measurement on  $\ket{\psi}\otimes \ket{\psi}$ and repeating $8n+1$ times, one can identify a set of $n+k$ independent generators with  probability at least $1-4^{-n}$.
\label{lemma:findbasis2} 
\el
\begin{proof}
	Suppose the expanded Pauli frame has $2k$ \pss\ $g_{n-k+1},\dots,g_n,h_{n-k+1},\dots,h_n$ and $n-k$ isotropic stabilizers $g_{1},\dots,g_{n-k}$.
	Since the joint Bell measurement on $\ket{\psi}\otimes \ket{\psi}$ does not reveal the correct phase of a stabilizer,
	we do not have to discuss this phase. 
	The Bell measurement returns an outcome corresponding to a Pauli operator in $ \sigma_{\bfr_0}\times \langle g_1,\dots,g_{n-k} \rangle \times \{I^{\otimes n},g_{n-k+1},h_{n-k+1}\}\times  \{I^{\otimes n},g_{n-k+2},h_{n-k+2}\} \times \cdots \times  \{I^{\otimes n},g_{n},h_{n}\}.$
	Setting the first outcome as a reference, the remaining $8n$ samples lie in 
	\begin{align*}
		&\langle g_1,\dots,g_{n-k} \rangle \times \{I^{\otimes n},g_{n-k+1},h_{n-k+1}\}\\
		&\times  \{I^{\otimes n},g_{n-k+2},h_{n-k+2}\} \times \cdots \times  \{I^{\otimes n},g_{n},h_{n}\}.
	\end{align*}    
	These outcomes have a total of  $3^k\times 2^{n-k}$ possibilities,
	of which $3^{k-1}\times 2^{n-k}$ outcomes are related to $h_n$, that is, each of the outcome is a product of $h_n$ and some other generators, such as $h_n,h_{n}g_{n-1},h_{n}h_{n-2},\dots,h_{n}h_{n-1}\cdots h_{n-k+1}$.
	The probability of  obtaining these outcomes in the modified outcome  is 
	\[
	\sum_{j=0}^{n-1} 2^i{n-1 \choose i}\frac{1}{2^{n+1+i}}=\frac{1}{4}.
	\]
	Thus the probability that $8n$ Bell measurement outcomes are not related to $h_n$ is $(3/4)^{8n}$.
	Since there are  $3^k\times 2^{n-k}\leq 3^n$ distinct Pauli outcomes and the probability of obtaining each outcome is no larger than $1/4$, 
	the probability that the $8n$ outcomes do not have $n+k$ independent generators is upper bounded by
	\begin{align}
	3^n\times  (3/4)^{8n}\leq 4^{-n}.
	\end{align}
\end{proof}

From Lemma~\ref{lemma:findbasis2}, one can find a generating set that contains the $2k$ \pss\ and $n-k$ \iss.
Unfortunately, we do not know any efficient method to find these $2k$ \pss\ and $n-k$ \iss.
The number of possibilities of $2k$ {\pss} is bounded by ${3^{k}\choose {2k}} <3^{2k^2}$.
This seems to suggest that we can handle  only $O(\sqrt{\log n})$ $T$ gates by brute force.
However, in the following we will show that $k=O({\log n})$ $T$ gates can be handled.

Before we show how to determine whether a set of $2k$ symplectic partners are our target \pss, we need the following lemma, which 
generalizes the idea of measuring a Pauli operator on a stabilizer state.
\bl
Suppose $\ket{\psi}\bra{\psi}= \sum_{\bfj\in\{1,2,3\}^k} \beta_\bfj \ket{\phi_\bfj}\bra{\phi_\bfj}$ is a  {full-rank} $T$-depth one  output state as defined in~Eq.\,(\ref{eq:stabilizer_pseduomixture}). Then
\begin{enumerate}
	\item If $g\in\cP_n$ is not a stabilizer of any of the $\ket{\phi_\bfj}\bra{\phi_\bfj}$,
	then measuring $g$ on $\ket{\psi}$ returns outcome $+1$ ($-1$) with probability $1/2$ $(1/2)$.
	\item If $g$ is an isotropic stabilizer, then  measuring $g$ on $\ket{\psi}$ returns outcome $+1$   with probability $1$.
	\item If $g$ is a \ps, then  measuring $g$ on $\ket{\psi}$ returns outcome $+1$ with   probability $\frac{2+\sqrt{2}}{4}$ and outcome $-1$ with probability $\frac{2-\sqrt{2}}{4}$.
	
	\item If $g$ is a product of $m$ primary symplectic stabilizers, 
	then  measuring $g$ on $\ket{\psi}$ returns outcome $+1$ with   probability $\frac{1}{2}(1+ 2^{-m/2})$ and outcome $-1$ with  probability $\frac{1}{2}(1- {2^{-m/2}})$.

\end{enumerate}

\label{lemma:Pauli_measurement_T}
\el
\begin{proof}
	1) 
		Consider  $g\in\cP_n$ but $g$ does not stabilize any of the $\ket{\phi_\bfj}$.
	For each $\ket{\phi_\bfj}$, there exists $g_\bfj$ such that $gg_\bfj=-g_\bfj g$. Then measuring $g$ on $\ket{\psi}\bra{\psi}$ 
	returns outcome $+1$ with probability
	\begin{align*}
		\Pr(+1)=&\tr\left( \frac{I^{\otimes n}+g }{2}\ket{\psi}\bra{\psi} \right)\\
		=&\sum_{\bfj=1}  \beta_\bfj  \tr\left( \frac{I^{\otimes n}+g }{2}\ket{\phi_\bfj}\bra{\phi_\bfj} \right)\\
		=&\sum_{\bfj} \beta_\bfj  \tr\left( \frac{I^{\otimes n}+g }{2}g_\bfj\ket{\phi_\bfj}\bra{\phi_\bfj} \right)\\
		=&\sum_{\bfj} \beta_\bfj  \tr\left( \frac{I^{\otimes n}-g }{2}\ket{\phi_\bfj}\bra{\phi_\bfj} \right),
	\end{align*}
	which is, by definition, $\Pr(-1)$. Thus $\Pr(+1)=\Pr(-1)=1/2$.
	
	2) 
	Consider $g\in\cP_n$ and $g$ stabilizes all of the $\ket{\phi_\bfj}$.
	Thus 
	$
	\Pr(+1) = \sum_{\bfj} \beta_\bfj  \tr\left( \frac{I+g }{2}\ket{\phi_\bfj}\bra{\phi_\bfj} \right)=1.
	$
	If $-g$ stabilizes all of the $\ket{\phi_\bfj}$, then 
	measuring $g$ on $\ket{\psi}\bra{\psi}$ 
	returns outcome $-1$ with probability one and we know that $-g$ is a stabilizer.
	
	3) 
	According to the assumption, we may say that $g$ is a primary symplectic partner, say $g=h_{n-k+1}$, that stabilizes exactly $1/3$ of the states $\{\ket{\phi_\bfj}\}$,  $g$'s symplectic partner, say $g_{n-k+1}$,  will stabilize another $1/3$ of  $\{\ket{\phi_\bfj}\}$,
	and   $-g_{n-k+1}$, will stabilize the other $1/3$ of  $\{\ket{\phi_\bfj}\}$.
	Then 
	\begin{align*}
		&\Pr(+1)\\
		=& \sum_{\bfj: \bfj_1=3  }\beta_\bfj  \tr\left( \frac{I+g }{2}\ket{\phi_\bfj}\bra{\phi_\bfj} \right)+ \sum_{\bfj: \bfj_1=1 }\beta_\bfj  \tr\left( \frac{I+g }{2}\ket{\phi_\bfj}\bra{\phi_\bfj} \right)\\
		&+ \sum_{\bfj: \bfj_1=2}\beta_\bfj  \tr\left( \frac{I+g }{2}\ket{\phi_\bfj}\bra{\phi_\bfj} \right)\\
		=&  \alpha_3 +\frac{1}{2}(1-\alpha_3)=\alpha_1 +\frac{1}{2}\alpha_3= \frac{2+\sqrt{2}}{4}.
	\end{align*}

	4)
	Similarly, 	
	\begin{align*}
		\Pr(+1)
		=&  \alpha_3^m+\frac{1}{2}(1-\alpha_3^m)=\frac{1}{2}(1+\alpha_3^m).
	\end{align*}
\end{proof}

Therefore, by measuring a Pauli operator $O(1/\epsilon)$ times for some accuracy $\epsilon>0$,
we can determine whether it is an isotropic stabilizer or a primary stabilizer from the distribution of the measurement outcomes.

Given  $n-k$ isotropic stabilizers $g_{1},\dots,g_{n-k}$ and $2k$ \pss\ $g_{n-k+1},\dots,g_n,$ $h_{n-k+1},\dots,h_n\in\cP_n$   satisfying the commutation relations (\ref{eq:commutation}),
one can find an effective circuit decomposition that generates the corresponding $T$-depth one output state. 
The complete algorithm for learning  the output state of a $T$-depth one quantum circuit  {of full $T$-rank} on input $ \ket{0^n}$ is given in Algorithm~\ref{Algorithm:II} below.
Note that Step (6) follows from Table~\ref{table:expandedT}. In summary, we have the following theorem of learning unknown $T$-depth one output states.

\begin{algorithm}
	
	\Input{$O(n)$ $|0^n\rangle$ states, oracle $U$, which is of $T$-depth one and full $T$-rank.}
	\Output{ {a quantum circuit $\hat{U}$ such that   $\hat{U}\ket{\bfv}=U\ket{\bfv}$ for $\bfv\in\{0,1\}^n$.}}

	\setcounter{AlgoLine}{0}
	
	\begin{enumerate}[(1)]
		\item Prepare $O(3^k n)$ copies of the unknown stabilizer state $\ket{\psi}\in \mC^{2^n}$ by querying to $U$.
		\item Do $8n+1$ Bell measurements on  $\ket{\psi}\otimes \ket{\psi}$ and let the outcomes be $\bfr_i'\in\{0,1\}^{2n}$ for $i=0,\dots, 8n$.
		Let $\bfr_i=\bfr_i' \bfr_0'$ for $i=1,\dots,8n$.
		
		\item Do a Gaussian elimination on $\sigma_{\bfr_i}$ to find $k+n$ independent generators, say $\{g_1',\dots,g_n',h_{n-k+1}',\dots,h_n'\}$,  satisfying the commutation relations (\ref{eq:commutation}). 
		Then $\{g_1',\dots,g_{n-k}'\}$ are a set of independent generators for the isotropic stabilizers of the underlying expanded Pauli frame.
		\item For each operator $g$ in   $\{I^{\otimes n}, g_{n-k+1},h_{n-k+1}\}\times \{I^{\otimes n},g_{n-k+2},h_{n-k+2}\}\times \cdots \times\{I^{\otimes n},g_n,h_n\}$, do the following:
		\begin{itemize}
			\item Measure $g$ on $\ket{\psi}$  $3200n$ times and obtain a probability distribution.
			\item Use Lemma~\ref{lemma:Pauli_measurement_T} to determine whether $g$ is a \ps.
			\item If a total of $2k$ \pss\ are found, halt.
		\end{itemize}
		We end up with $2k$ \pss\ $\{g_{n-k+1},h_{n-k+1},\dots,g_n,h_n \}$.
		\item Use Lemma~\ref{lemma:circuit_synthesis} to find a Clifford circuit $C$ such that $CX_iC^\dag= g_i$ for $i=1,\dots,{n-k}$ and $\{CX_jC^\dag,CY_jC^\dag\} = \{g_j, h_j\},\{g_j, -h_j\},\{-g_j, h_j\},$ or $\{-g_j, -h_j\} $ for $j={n-k+1},\dots,n$.
		\item  Let $\bfs\in\{0,1,2,3\}^n$ with $\bfs_i=0$ for $i=1,\dots, n-k$.  For $i={n-k+1},\dots,n$,
		\begin{itemize}
			\item if	$\{C^\dag g_i C,C^\dag h_i C\} = \{X_i, Y_i\} $, $\bfs_i=0$;
			\item if	$\{C^\dag g_i C,C^\dag h_i C\} = \{-X_i, Y_i\} $, $\bfs_i=1$;
			\item if	$\{C^\dag g_i C,C^\dag h_i C\} = \{-X_i, -Y_i\} $, $\bfs_i=2$;
			\item if	$\{C^\dag g_i C,C^\dag h_i C\} = \{X_i, -Y_i\} $, $\bfs_i=3$;
		\end{itemize}

		\item Output the circuit $\hat{U}=C\circ \left({I}^{\otimes n-k}\otimes T^{\otimes k}\right)\circ S^\bfs \circ H^{\otimes n} $.

	\end{enumerate}

	\caption{ {Learning  a $T$-depth one quantum circuit $U$  {of full $T$-rank} on the computational basis}} \label{Algorithm:II}
\end{algorithm}

\begin{theorem} \label{theo:T}
	Given access to an unknown $T$-depth one quantum circuit $U$  {of full $T$-rank}, one can learn a circuit description using  $O(3^k n)$ queries to the unknown circuit $U$ with time complexity  {$O(n^3+3^kn)$}, where $k\leq n$ is the number of $T$ gates, so that 	the produced hypothesis circuit $\hat{U}$ is equivalent to $U$ with probability at least
		$1-3e^{-n}$ when the input states are restricted to the computational basis. That is,
		\begin{align*}
			\hat{U}\ket{\bfv} =U\ket{\bfv}, \ \bfv\in\{0,1\}^n.
		\end{align*}
\end{theorem}
\begin{proof}
	By  Lemma~\ref{lemma:findbasis2}, we know that step (3) of Algorithm~\ref{Algorithm:II} fails with probability at most $4^{-n}$.
	In step (4) of Algorithm~\ref{Algorithm:II}, we have to determine whether an operator $g$ is a \ps\ or not. According to Lemma~\ref{lemma:Pauli_measurement_T}, we have to distinguish the probability distribution of measuring a \ps\ from the other three cases. 
	Assume that $g$ is a \ps\ for simplicity.
	The  distribution of measuring the product of two \pss\ has mean $\frac{3}{4}$ so that these two distributions have the smallest gap of $\delta=\frac{\sqrt{2}-1}{4}$. Let $R$ denote the average of the $3200n$ measurement outcomes of $g$ (see Step (4) of Algorithm~\ref{Algorithm:II}), hence $\E\{R\}=\frac{2+\sqrt{2}}{4}$ when $g$ is a \ps.
	Using the standard Chernoff bound, 
	\[
	\Pr\{ |R-\E\{R\}| \geq \delta/2\} \leq 2 e^{-3200n\times \delta^2/16}=2 e^{-2.14n}.
	\]
	We repeat this process for at most $3^k$ times. By the union bound, the error rate of step (4) is at most 
	\[
	3^k 2 e^{-2.14n}\leq 2 e^{-n}.
	\]	
	Therefore the total error rate is bounded by $4^{-n}+2 e^{-n} \leq 3e^{-n}$.

		Finally, it is clear that we have $\hat{U}\ket{0^n} =U\ket{0^n}$ by Algorithm~\ref{Algorithm:II}.
		In fact, for an arbitrary input state $\ket{\bfv}$ for $\bfv\in\{0,1\}^n$ in the computational basis,
		its stabilizers are $\{(-1)^{\bfv_i} Z_i, i=1,\dots,n\}$. Thus $U\ket{\bfv}$ and $U\ket{0^n}$  have the same {\iss} and {\pss} up to phases $-1$, and  $\hat{U}\ket{\bfv}$ and $\hat{U}\ket{0^n}$  have the same {\iss} and {\pss} up to phases $-1$. Since the evolution of $Z_i$ under $U$ is identical to that under $\hat{U}$ and the initial phases are unaffected in the evolution,
		we have  $\hat{U}\ket{\bfv}= U\ket{\bfv}$ for any $\bfv\in\{0,1\}^n$.
\end{proof}
 
	Next we comment on the task of learning an unknown quantum unitary $U$ that may apply on arbitrary input vectors.
	\bc \label{cor:tcircuit}
	Given access to an unknown $T$-depth one quantum circuit $U$  {of full $T$-rank}, one can learn a circuit description   using  $O(3^k n)$ queries to the unknown circuit $U$ with time complexity  $O(4^{3n})$, where $k\leq n$ is the number of $T$ gates,   so that the produced hypothesis circuit $\hat{U}$ is equivalent to $U$ up to a phase.
	\ec
	\begin{proof}		
		From the previous theorem, we are able to find a circuit $\hat{U}$ such that $\hat{U}\ket{\bfv}=U\ket{\bfv}$ for any $\bfv\in\{0,1\}^n$, using $O(3^{k} n)$ queries to $U$. 
		
		Now we can use $O(4^n)$ classical computations to compute  $U\ket{\bfv}\bra{\bfw}U^\dag$ for any $\bfv,\bfw\in\{0,1\}^n$. Consequently, by Lemma~\ref{lemma:2n}, $U$ can be uniquely determined using   $O(3^{k} n)$ queries to $U$ with   $O(4^n)$ classical computations. 
		 {Then we may use a Gaussian elimination in time $O(4^{3n})$ to find a circuit description.}
	\end{proof}

\be
Consider the following $T$-depth one circuit: 
\begin{figure}[h]
	\[
	\includegraphics[width=0.5\linewidth]{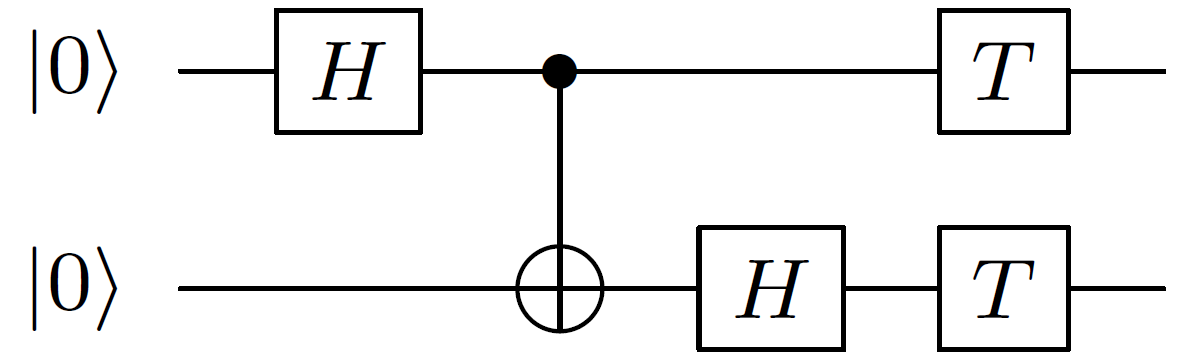}
	\]
 
\end{figure}

The output state $\ket{\psi}$ has an expanded Pauli frame generated by \pss\ $\{ZX, ZY, XZ,YZ\}$.
Possible Bell measurement outcomes on $\ket{\psi^*}\otimes\ket{\psi}$ and their respective Pauli measurement (with outcome $+1$) probabilities from Lemma~\ref{lemma:BellMeasure2} and Lemma~\ref{lemma:Pauli_measurement_T} are given in Table: 
\begin{table}[h]
		\resizebox{\columnwidth}{!}{%
		\begin{tabular}{|c|c|c|c|c|c|c|c|c|c|} %
			\hline
			$g$	& $II$ & $XX$&$XY$ &$XZ$& $YX$ &$YY$ &$YZ$ &$ZX$ &$ZY$\\
			\hline
			Bell meas.  & 1/4&1/16&1/16&1/8&1/16&1/16&1/8&1/8&1/8\\
			\hline
			Pauli meas. & $1/2$ & $3/4$ &$3/4$ & $\frac{2+\sqrt{2}}{4}$ & 3/4 &3/4 & $\frac{2+\sqrt{2}}{4}$ &$\frac{2+\sqrt{2}}{4}$ &$\frac{2+\sqrt{2}}{4}$ \\
			\hline
		\end{tabular}
}
\end{table} 

Similarly,  possible Bell measurement outcomes on $\ket{\psi}\otimes\ket{\psi}$  are given in the following table:
Since $(ZY)^*= -ZY$ and  $(YZ)^*= -YZ$, by Lemma~\ref{lemma:BellMeasure3} and Eq.\,(\ref{eq:sigma_r0}), the measurement outcomes are shifted by $ZX\cdot XZ =YY$, as can be seen from the two tables.
\begin{table}[h]
		\resizebox{\columnwidth}{!}{%
		\begin{tabular}{|c|c|c|c|c|c|c|c|c|c|} %
			\hline
			$g$	& $YY$ & $ZZ$& $ZI$ &$ZX$& $IZ$ &$II$ &$IX$ &$XZ$ &$XI$\\
			\hline
			 Bell meas.  & 1/4&1/16&1/16&1/8&1/16&1/16&1/8&1/8&1/8\\
			\hline
			  Pauli meas.  & $1/2$ & $3/4$ &$3/4$ & $\frac{2+\sqrt{2}}{4}$ & 3/4 &3/4 & $\frac{2+\sqrt{2}}{4}$ &$\frac{2+\sqrt{2}}{4}$ &$\frac{2+\sqrt{2}}{4}$ \\
			\hline
		\end{tabular}
	}
\end{table}

Assume that after performing Bell measurements  on $\ket{\psi}\otimes\ket{\psi}$ for many times, one determines two pairs of symplectic partners $\{YX,YZ\}$ and   $\{XY,ZY\}$. 
Thus we check all the operators in 
\begin{align*}
	&\{II,YX,YZ \}\times \{II,XY,ZY\}\\
	=&\{II, YX,YZ, XY,ZZ,ZX,ZY,XZ,XX\}
\end{align*}
and we can determine that $\{ZX,ZY\}$ and   $\{XZ,YZ\}$ are \pss\ by using Pauli measurements as in Lemma~\ref{lemma:Pauli_measurement_T}. 
The remaining steps are straightforward and are omitted. The output circuit is shown in the following figure. 
\begin{figure}[h]
	\[
	\includegraphics[width=0.95\linewidth]{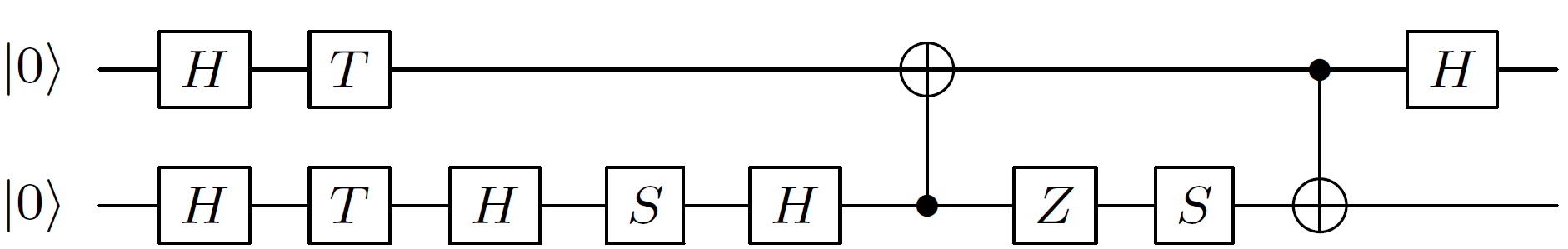}
	\]
\end{figure}

\eep

\be
Consider a $T$-depth one output state defined by \pss\ $g_1=YZ, g_2=ZY, h_1=-YX,h_2=-XY$,
where $g_i^*=- g_i$ and $h_j^*=-h_j$.
By Lemma~\ref{lemma:BellMeasure3} and Eq.\,(\ref{eq:sigma_r0}), 
the outcome distribution of Bell measurements on $\ket{\psi}\otimes \ket{\psi}$ is shifted by $\sigma_{\bfr_0}=g_1g_2h_1h_2=YY$ from the 
outcome distribution of Bell measurements on $\ket{\psi^*}\otimes \ket{\psi}$. 
\eep

In fact, the output state of a circuit of $T$-depth $n$ on input $\ket{0^n}$ may be learned by Algorithm~\ref{Algorithm:I}
as long as it is equivalent to a $T$-depth $1$ output state.
\be
The following circuit is of $T$-depth 2 
and the circuit output on input $\ket{0^n}$ can be learned by Algorithm~\ref{Algorithm:I}.
\begin{figure}[h]
	\[
	\includegraphics[width=0.65\linewidth]{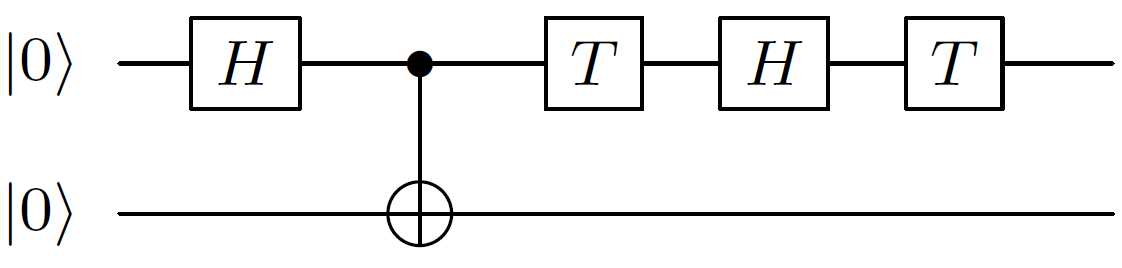}
	\]
\end{figure}
This can be understood as the circuit on input $\ket{00}$ is equivalent to the circuit in the previous example. 
\eep

\section{Discussions} \label{sec:conclusions}

In this paper, we studied the problem of learning unknown quantum circuits composed of Clifford circuits possibly combined with some $T$ gates.
When the underlying target circuit is an $n$-qubit Clifford $\mathsf{C}_n$, we can identify the unknown target by using $O(n^2)$ queries to $\mathsf{C}_n$.
The novelty of our result is that our algorithm exactly generates a circuit representation for the target,  rather than demonstrating the existence of a unique circuit corresponding to the query outcomes.
We emphasize this step is non-trivial since it requires a crucial method for circuit synthesis.
As a result, we are able  to predict the future output by sending an arbitrary input state to the proposed circuit.
Furthermore, our approach does not require querying the conjugate circuit $\mathsf{C}_n^\dagger$ since, to the best of our knowledge, no efficient way of generating such a conjugate circuit is known~\cite{MSM19}.

If the unknown target is a $T$-depth one quantum circuit of full $T$-rank and at most $O( \log n  )$ $T$ gates, we show that $O(n^2)$ queries to the circuit are sufficient to reconstruct the output state when the input is given by the all-zero state. The key ingredient is a novel expanded stabilizer formalism that allows us to trace the evolution of the expanded Pauli frame. Again, we do not require accessing to the conjugate circuit.

Whether any quantum circuit of $T$ depth-one (e.g.~with $k\leq n$ many $T$ gates) is efficiently learnable is still open. However, from our analysis, we do not expect that will be true.
Another interesting open problem is that how to obtain its circuit representation with an arbitrary  input state.
{One may also consider a formulation of PAC-learning these circuits on input the computational basis. The techniques of Bell measurements have to be modified appropriately.}

	Another potential direction is to use the stabilizer decomposition of a nonstabilizer state. In a classical simulation of quantum computation, one would like to minimize a possible representation of the quantum state and this leads to the use of a  representation with a low-stabilizer rank. 
	However, we do not know any measurement methods that can extract the component stabilizer states or their stabilizer generators. Even if one can retrieve the stabilizers using certain measurements, it is still difficult to regroup them so that we can identify each component stabilizer state because there is no constraint on the component stabilizer states.
	On the contrary, the stabilizer pseudomixture of a full-rank $T$-depth one circuit follows the algebraic structure we derived in the manuscript. This allows us to use Bell sampling to extract and regroup the stabilizers.
	Nonetheless, if one can develop a learning procedure with respect to the stabilizer decomposition, it would shed light on learning any nonstabilizer states.

	Finally, one may consider the applications of our learning algorithms to near term noisy quantum computers.
	Unfortunately, we believe that learning a noisy Clifford circuit would be computationally hard.
	That said, if learning a noisy Clifford circuit is not hard, then it implies that learning a noisy CNOT circuit is not hard, which in turn, implies that learning a parity function with noise (LPN) is not hard.
	However, we know that is not the case; LPN is believed to be hard~\cite{BKW03}.
	This does not mean that learning quantum states/circuits is not useful, but it reveals the nature that learning quantum circuits (so as learning certain classical circuits) is a highly technical demanding task.

\section*{ACKNOWLEDGEMENT}
CYL thanked Kai-Min Chung, Kao-Yueh Kuo, and Yingkai Ouyang  for helpful discussions.


\ifCLASSOPTIONcaptionsoff
  \newpage
\fi

\end{document}